\documentclass[sigconf]{acmart}

\usepackage{graphicx}
\usepackage{picinpar}
\usepackage{pifont}
\DeclareGraphicsExtensions{.eps,.pdf,.png,.jpg}
\graphicspath{{figure/}}
\usepackage{subfigure}
\usepackage{algpseudocode}  
\usepackage{pifont}         
\usepackage[ruled,linesnumbered,vlined]{algorithm2e}
\usepackage{makecell}
\usepackage{xcolor}
\usepackage{geometry}
\usepackage{caption}
\usepackage{threeparttable}
\usepackage{array}
\usepackage{amsmath}
\usepackage{amsmath}
\usepackage{latexsym}
\usepackage{amsthm}

\usepackage{booktabs}
\usepackage{multirow}
\usepackage{tabularx}
\usepackage{colortbl}

\usepackage{tikz}
\usetikzlibrary{shapes.geometric, arrows.meta, positioning, fit, backgrounds, calc, patterns}

\usepackage{xcolor}
\definecolor{queuecolor}{RGB}{66, 133, 244}
\definecolor{adjcolor}{RGB}{52, 168, 83}
\definecolor{cachecolor}{RGB}{251, 188, 4}
\definecolor{memorycolor}{RGB}{234, 67, 53}
\definecolor{codebg}{RGB}{245, 245, 245}

\usepackage{listings}
\usepackage{url}
\usepackage[normalem]{ulem}
\usepackage{xspace}
\usepackage{float}
\usepackage{caption}
\usepackage{subcaption}

\renewcommand\footnotetextcopyrightpermission[1]{}

\newtheorem{property}{Property}
\renewcommand{\na}{---}
\renewcommand{\shortauthors}{}
\begin{document}
\title{StreamTGN: A GPU-Efficient Serving System for Streaming Temporal Graph Neural Networks}
\author{Lingling Zhang}
\affiliation{\institution{Capital Normal University}\country{}}
\email{7089@cnu.edu.cn}

\author{Pengpeng Qiao}
\affiliation{\institution{Institute of Science Tokyo}\country{}}
\email{peng2qiao@gmail.com}

\author{Zhiwei Zhang}
\affiliation{\institution{Beijing Institute of Technology}\country{}}
\email{zwzhang@bit.edu.cn}

\author{Ye Yuan}
\affiliation{\institution{Beijing Institute of Technology}\country{}}
\email{yuan-ye@bit.edu.cn}

\author{Guoren Wang}
\affiliation{\institution{Beijing Institute of Technology}\country{}}
\email{wanggr@bit.edu.cn}
\begin{abstract}
Temporal Graph Neural Networks (TGNs) achieve state-of-the-art
performance on dynamic graph tasks, yet existing systems focus
exclusively on accelerating training---at inference time, every
new edge triggers $O(|V|)$ embedding updates even though only a
small fraction of nodes are affected. We present
\textbf{StreamTGN}, the first streaming TGN inference system
exploiting the inherent locality of temporal graph updates: in an
$L$-layer TGN, a new edge affects only nodes within $L$ hops of
the endpoints, typically less than 0.2\% on million-node graphs.
StreamTGN maintains persistent GPU-resident node memory and uses
dirty-flag propagation to identify the affected set $\mathcal{A}$,
reducing per-batch complexity from $O(|V|)$ to $O(|\mathcal{A}|)$
with zero accuracy loss. Drift-aware adaptive rebuild scheduling
and batched streaming with relaxed ordering further maximize
throughput. Experiments on eight temporal graphs (2K--2.6M nodes)
show 4.5$\times$--739$\times$ speedup for TGN and up to
4,207$\times$ for TGAT, with identical accuracy. StreamTGN is
orthogonal to training optimizations: combining SWIFT with
StreamTGN yields 24$\times$ end-to-end speedup across three
architectures (TGN, TGAT, DySAT).
\end{abstract}
\maketitle
\section{Introduction}
Temporal Graph Neural Networks (TGNs) have achieved state-of-the-art performance on dynamic graph tasks such as link prediction \cite{xu2025unidyg} and node classification~\cite{rossi2020temporal,xu2020inductive,kumar2019predicting}.
However, existing TGN systems process graphs in batch mode, recomputing embeddings for all nodes whenever the graph changes \cite{li2023zebra}. This approach is fundamentally inefficient for streaming applications where edges arrive continuously---a single new edge triggers $O(|V|)$ recomputation even though only a small fraction of nodes are actually affected.
 
A growing body of work has sought to accelerate TGN training. TGL~\cite{zhou2022tgl, chen2023neutronstream, wang2024tglite} designed a unified framework with parallel temporal sampling for large-scale training. ETC~\cite{gao2024etc} introduced adaptive batching and a three-step data access policy to reduce redundant data transfer. SIMPLE~\cite{gao2024simple} proposed dynamic GPU data placement to alleviate the CPU--GPU loading bottleneck. SWIFT~\cite{guo2025swift} developed a secondary-memory pipeline to distribute data efficiently across GPU, main memory, and disk. While these systems substantially reduce training time, they all share the same limitation: \emph{at inference time, they execute the same full-recomputation pipeline as TGL}, because their optimizations
target backward-pass scheduling, data loading, or batch construction, none of which apply during online serving \cite{sun2025hyperion}. As illustrated in Figure~\ref{StreamGCN}, even a modest reduction in per-query inference latency has massive cumulative impact in
production systems. 
 
\begin{figure}[hbt!]
\centering
\includegraphics[width=0.46\textwidth]{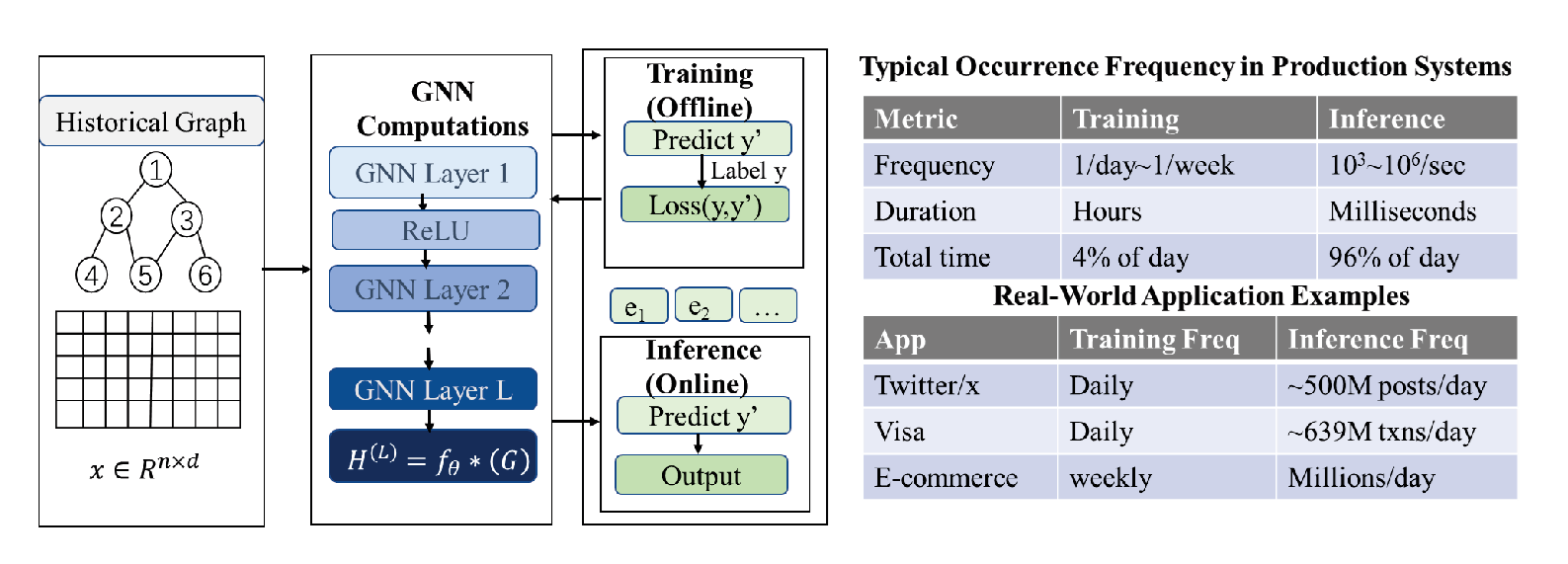}
\caption{Overview of TGN training and inference. Training runs offline and infrequently; inference runs continuously at scale. Even a small improvement in inference latency yields enormous savings: $(10\text{\,ms} - 5\text{\,ms}) \times 10^8
\text{ queries/day} \approx 10^6 \text{ seconds/day}$.}
\label{StreamGCN}
\end{figure}
 
We make a key observation: the influence of a new edge
is inherently local.
In an $L$-layer TGN, only nodes within $L$ hops of
the newly arrived edge have their embeddings changed;
all other node embeddings remain mathematically
identical to their previous values.
On sparse real-world graphs, this \emph{affected set}
$\mathcal{A}$ typically constitutes less than 1--5\%
of the total node population---and on large-scale
graphs with millions of nodes, the fraction drops
below 0.2\%.
This locality creates a massive optimization
opportunity, particularly because inference dominates
the computational budget in production systems.
As illustrated in Figure~\ref{StreamGCN}, training
runs offline and infrequently (daily or weekly),
consuming only ${\sim}$4\% of total compute time,
while inference runs continuously at
$10^3$--$10^6$ queries per second, accounting for
${\sim}$96\% of compute time.
Even a modest 5\,ms reduction in per-query latency
translates to ${\sim}10^6$ seconds of saved compute
per day at $10^8$ daily queries---yet existing
systems (TGL, ETC, SIMPLE, SWIFT) optimize
exclusively for the 4\%, leaving the dominant 96\%
untouched.

Consider, for example, a financial institution
processing 10,000 transactions per second across 10
million accounts.
To detect fraud in real time, the system must evaluate
each transaction within milliseconds.
A traditional TGN implementation recomputes embeddings
for all 10 million accounts per
transaction---requiring ${\sim}$100\,ms and achieving
only 10 transactions per second, far below the
required throughput.
Yet each transaction affects only the accounts within
$L$ hops of the involved parties---typically fewer
than 0.001\% of all accounts.
 
This observation motivates \textbf{StreamTGN}, the first streaming TGN inference system that achieves provably optimal incremental updates, recomputing only the mathematically necessary nodes while guaranteeing identical results to full recomputation. StreamTGN makes three technical contributions:
 
\begin{itemize}
\item \textbf{Persistent GPU-resident memory with dirty-flag
tracking.}
StreamTGN maintains node memory states persistently on the GPU
across batches. When new edges arrive, a lightweight propagation mechanism marks
only the affected nodes---those within $L$ hops of the updated
edges as dirty. Inference then recomputes embeddings exclusively for dirty nodes, reusing cached embeddings for the remainder. This reduces per-batch inference complexity from $O(|V|)$ to $O(|\mathcal{A}|)$, where $|\mathcal{A}| \ll |V|$ on large graphs.
 
\item \textbf{Drift-aware adaptive rebuild scheduling.}
Over time, accumulated incremental updates may cause cached embeddings to drift from their full-recomputation values due to
higher-order neighborhood effects. StreamTGN monitors the divergence between incremental and full-rebuild outputs and triggers a partial or full rebuild at
the optimal moment.
Our experiments show that this adaptive strategy reduces rebuild
cost by an order of magnitude compared to naive periodic
rebuilding, while maintaining prediction accuracy within a
provable bound.
 
\item \textbf{Batched streaming with relaxed ordering.}
Strict sequential processing of one edge at a time severely limits
throughput.
StreamTGN groups incoming edges into batches and processes them in
parallel under a bounded-staleness model: edges within the same
batch share a logical timestamp, and the system guarantees that
the resulting embeddings differ from the strictly sequential output
by at most $\delta$.
This relaxation enables an order-of-magnitude throughput improvement
over strict sequential processing while preserving accuracy within
a provable bound.
\end{itemize}
 
We evaluate StreamTGN on eight real-world temporal graphs spanning
four orders of magnitude in scale.
On the streaming inference task, StreamTGN achieves
4.5$\times$--739$\times$ speedup over TGL for the TGN model, and up
to 4,207$\times$ for the non-memory TGAT model, with zero accuracy
degradation.
When combined with the best training-phase system (SWIFT for training,
StreamTGN for inference), the end-to-end pipeline achieves up to
24$\times$ total speedup on Stack-Overflow.
Crucially, StreamTGN is \emph{orthogonal} to all existing
training-phase optimizations: it can be deployed on top of any
training system to accelerate the inference phase that none of them
address.
\section{Preliminaries}
\label{sec:background}
This section introduces the foundations of Temporal Graph Neural Networks (TGNNs). We first define temporal graphs, then formulate the TGNN problem, describe the architecture, present the mathematical formulation, discuss training and inference, and finally identify the sequential bottleneck that motivates our work.
\subsection{Temporal Graph Definition}
\label{subsec:temporal_graph}
\begin{definition}[Temporal Graph]
A temporal graph denoted $G^t$ is defined as a sequence of time-stamped edges: $$G^t = \{(u_1, v_1, t_1, \mathbf{e}_1), (u_2, v_2, t_2, \mathbf{e}_2), \ldots, (u_E, v_E, t_E, \mathbf{e}_E)\}$$ where each tuple $(u_i, v_i, t_i, \mathbf{e}_i)$ represents an interaction 
between two nodes $u_i$ and $v_i$ at timestamp $t_i$, 
with edge features $\mathbf{e}_i \in \mathbb{R}^{d_e}$. The edges are ordered 
chronologically: $t_1 \leq t_2 \leq \cdots \leq t_E$.
\end{definition}
\begin{definition}[Node Set and Features]
The node set $V = \{v_1, v_2, \ldots, v_N\}$ contains $N$ nodes. Each node $v$ may have static features $\mathbf{x}_v \in \mathbb{R}^{d_n}$.
\end{definition}
\begin{definition}[Temporal Neighborhood]
The temporal neighborhood of node $v$ before time $t$ is defined as: $$\mathcal{N}_v^{<t} = \{(u, t', \mathbf{e}) : (u, v, t', \mathbf{e}) \in G \land t' < t\}$$.This captures all historical interactions involving node $v$ up to (but not including) time $t$.
\end{definition}

A Temporal Graph Neural Network (TGNN) learns to generate node representations that capture both structural and temporal patterns \cite{xu2024timesgn, huang2023temporal}, thereby supporting downstream tasks such as node classification and link prediction \cite{liu2023stgin, sheng2024mspipe}.
\begin{definition}[TGNN Learning Problem]
Given a temporal graph $G^t$ containing all edges up to time $t$, a Temporal Graph Neural Network (TGNN) learns a mapping function $\Phi: G^t \rightarrow H^t$ that transforms the temporal graph into a set of node embeddings $H^t = \{\mathbf{h}_v(t) : v \in V\}$, where $\mathbf{h}_v(t) \in \mathbb{R}^{d_h}$ denotes the embedding vector of node $v$ at time $t$. where $\mathbb{R}^{d_h}$ represents the $d_h$-dimensional real-valued vector space, and $d_h$ is the hidden dimension that determines the expressiveness of the learned representations \cite{wang2024efficient, xia2024redundancy}.
\end{definition}

\subsection{TGNN Architecture}
\label{subsec:tgnn_architecture}
As a variation of graph neural networks, TGNNs iteratively execute neural update and aggregation operations\cite{wang2021gnnadvisor, vatter2023evolution}, where the former primarily features dense matrix multiplication and the latter is characterized by the SpMM operation \cite{yang2022gnnlab}. Specifically, a TGNN performs these two operations over temporal graphs through a pipeline of five stages, as exemplified by TGN~\cite{rossi2020temporal, gravina2024deep}: time encoding, message computation, message aggregation, memory update, and embedding generation. Figure~\ref{fiveModules} illustrates the architecture. We explain the five modules along with the two types of operations as follows.
\begin{figure}[hbt!]
	\centering
	\includegraphics[width=0.46\textwidth]{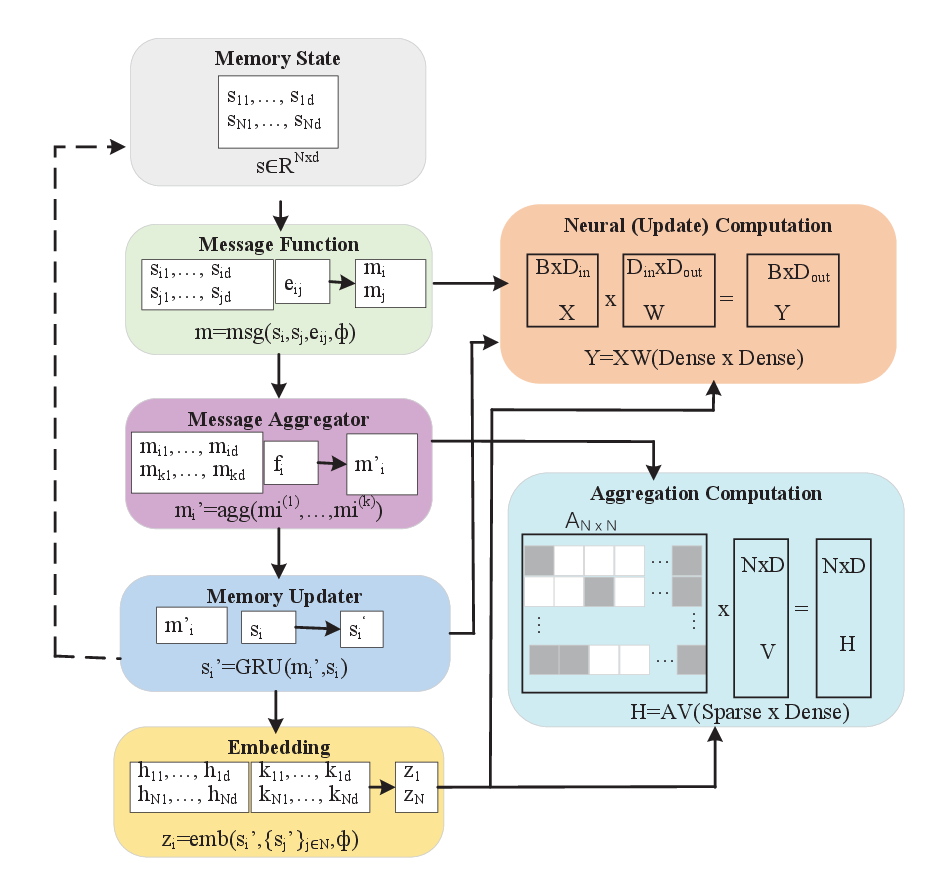}
	\caption{Overview of TGNN architecture including five modules with interleaved neural update and aggregation operations.}\label{fiveModules}
\end{figure}
\subsubsection{Stage 1: Time Encoding}
Before processing edges, TGN encodes temporal information using a learnable time encoding function that captures both short-term and long-term temporal patterns:
\begin{equation}
    \phi(t) = \sqrt{\frac{1}{d_t}} \left[ \cos(\omega_1 t), \sin(\omega_1 t), \ldots, \cos(\omega_{d_t/2} t), \sin(\omega_{d_t/2} t) \right]^\top
\end{equation}
where $\omega_i$ are learnable frequency parameters and $d_t$ is the time encoding dimension. This encoding is used throughout the subsequent stages to incorporate temporal information\cite{milani2019relationship}.
\subsubsection{Stage 2: Message Computation}
When a new edge $(u, v, t, \mathbf{e})$ arrives, the message function computes messages for both the source and destination nodes  (labeled as ${MSG}_s$ and ${MSG}_d$) based on their current memory states, edge features, and time information:
\begin{equation}
\begin{aligned}
    \mathbf{m}_u(t) &= \text{MSG}_s\left(\mathbf{s}_u(t^-), \mathbf{s}_v(t^-), \mathbf{e}, \phi(t - t_u^{\text{last}})\right) \\
    \mathbf{m}_v(t) &= \text{MSG}_d\left(\mathbf{s}_v(t^-), \mathbf{s}_u(t^-), \mathbf{e}, \phi(t - t_v^{\text{last}})\right)
\end{aligned}
\label{eq:msg}
\end{equation}
where $\mathbf{s}_u(t^-)$ denotes the memory state of node $u$ immediately before time $t$,  $t_u^{\text{last}}$ is the timestamp of node $u$'s most recent interaction, and $\phi(\cdot)$ is the time encoding function.
A common implementation concatenates the inputs and applies a linear transformation:
\begin{equation}
    \mathbf{m}_u(t) = W_m \cdot \left[\mathbf{s}_u(t^-) \| \mathbf{s}_v(t^-) \| \mathbf{e} \| \phi(\Delta t)\right] + \mathbf{b}_m
    \label{eq:msg_linear}
\end{equation}
where $\|$ denotes concatenation, $W_m \in \mathbb{R}^{d_m \times (2d_s + d_e + d_t)}$, and $\mathbf{b}_m \in \mathbb{R}^{d_m}$. This stage focuses on the neural update operation of dense matrix multiplication.

\subsubsection{Stage 3: Message Aggregation}
When multiple edges involving the same node arrive simultaneously (within a batch), the aggregator combines their messages into a single aggregated message:
\begin{equation}
    \bar{\mathbf{m}}_v(t) = \text{AGG}\left(\left\{\mathbf{m}_v^{(i)}(t) : i \in \mathcal{B}_v(t)\right\}\right)
    \label{eq:msg_agg}
\end{equation}
where $\mathcal{B}_v(t)$ is the set of edges involving node $v$ in the current batch. 
Common aggregation functions include:
\begin{equation}
\begin{aligned}
    \text{Mean:} \quad &\bar{\mathbf{m}}_v(t) = \frac{1}{|\mathcal{B}_v(t)|} \sum_{i \in \mathcal{B}_v(t)} \mathbf{m}_v^{(i)}(t) \\
    \text{Last:} \quad &\bar{\mathbf{m}}_v(t) = \mathbf{m}_v^{(\text{latest})}(t) \\
    \text{Sum:} \quad &\bar{\mathbf{m}}_v(t) = \sum_{i \in \mathcal{B}_v(t)} \mathbf{m}_v^{(i)}(t)
\end{aligned}
\label{eq:agg_functions}
\end{equation}
This stage focuses on the aggregation operation in terms of sparse-dense matrix multiplication.
\subsubsection{Stage 4: Memory Update}
The memory module maintains a memory state $\mathbf{s}_v(t) \in \mathbb{R}^{d_s}$ 
for each node $v$, serving as a compressed representation of node $v$'s 
historical interactions up to time $t$. The memory is initialized to zero, i.e., 
$\mathbf{s}_v(0) = \mathbf{0}, \forall v \in V$.
After message aggregation, the memory state is updated using a Gated Recurrent Unit (GRU):
\begin{equation}
    \mathbf{s}_v(t) = \text{GRU}\left(\bar{\mathbf{m}}_v(t), \mathbf{s}_v(t^-)\right)
    \label{eq:mem_update}
\end{equation}
The GRU computes:
\begin{equation}
\begin{aligned}
    \mathbf{z}_v &= \sigma\left(W_z \bar{\mathbf{m}}_v(t) + U_z \mathbf{s}_v(t^-) + \mathbf{b}_z\right) \\
    \mathbf{r}_v &= \sigma\left(W_r \bar{\mathbf{m}}_v(t) + U_r \mathbf{s}_v(t^-) + \mathbf{b}_r\right) \\
    \tilde{\mathbf{s}}_v &= \tanh\left(W_h \bar{\mathbf{m}}_v(t) + U_h \left(\mathbf{r}_v \odot \mathbf{s}_v(t^-)\right) + \mathbf{b}_h\right) \\
    \mathbf{s}_v(t) &= \left(1 - \mathbf{z}_v\right) \odot \tilde{\mathbf{s}}_v + \mathbf{z}_v \odot \mathbf{s}_v(t^-)
\end{aligned}
\label{eq:gru_detail}
\end{equation}
where $\mathbf{z}_v \in \mathbb{R}^{d_s}$ is the update gate controlling how much of 
the old state to retain, $\mathbf{r}_v \in \mathbb{R}^{d_s}$ is the reset gate controlling 
how much of the old state to use in computing the candidate, $\tilde{\mathbf{s}}_v$ is 
the candidate new state, $\sigma$ is the sigmoid function, and $\odot$ denotes 
element-wise multiplication. This stage focuses on the neural update operation of dense matrix multiplication.

\subsubsection{Stage 5: Embedding Generation}
The embedding module generates the final node embedding by aggregating information 
from the node's temporal neighborhood using multi-head temporal attention:
\begin{equation}
    \mathbf{h}_v(t) = \text{MultiHead}\left(\mathbf{q}_v, \{\mathbf{k}_u, \mathbf{v}_u : u \in \mathcal{N}_v^{<t}\}\right)
    \label{eq:emb_multihead}
\end{equation}
For each attention head, the query, key, and value are computed as:
\begin{equation}
\begin{aligned}
    \mathbf{q}_v &= \left[\mathbf{s}_v(t^-) \| \mathbf{x}_v \| \phi(0)\right] W_Q \\
    \mathbf{k}_u &= \left[\mathbf{s}_u(t_e^-) \| \mathbf{x}_u \| \mathbf{e}_{uv} \| \phi(t - t_e)\right] W_K \\
    \mathbf{v}_u &= \left[\mathbf{s}_u(t_e^-) \| \mathbf{x}_u \| \mathbf{e}_{uv} \| \phi(t - t_e)\right] W_V \\
    \alpha_{uv} &= \frac{\exp\left(\mathbf{q}_v^\top \mathbf{k}_u / \sqrt{d_k}\right)}{\sum_{w \in \mathcal{N}_v^{<t}} \exp\left(\mathbf{q}_v^\top \mathbf{k}_w / \sqrt{d_k}\right)} \\
    \mathbf{h}_v(t) &= \sum_{u \in \mathcal{N}_v^{<t}} \alpha_{uv} \mathbf{v}_u
\end{aligned}
\label{eq:attn_detail}
\end{equation}
where $t_e$ is the timestamp of edge $(u, v)$, $W_Q, W_K, W_V$ are learnable 
projection matrices, and $d_k$ is the dimensionality of the key vectors.
This stage involves both the neural update operation for computing queries, keys, and values via dense matrix multiplication, and the aggregation operation for the weighted summation over neighbors via sparse-dense matrix multiplication.

\subsection{Training and Inference Flow}
A TGNN operates through two tasks: training and inference \cite{xu2024timesgn, zhang2023tiger}, both of which execute the same five-stage pipeline described above. The key distinction lies in their time spans, processed data, and computational requirements.

During \textbf{training}, the model processes historical edges over a long time span covering the majority of the temporal graph (e.g., 70\% for training and 15\% for validation) to learn model parameters \cite{zhou2023disttgl, yang2019aligraph}. Edges are processed in temporal order within each batch: (1) embeddings are generated using the current memory states, (2) predictions are made and loss is computed, (3) gradients are backpropagated to update model parameters, (4) messages are computed and aggregated, and (5) memory states are updated for the next batch. The key constraint is that memory updates must respect temporal causality—a node's memory at time $t$ can only depend on interactions before $t$ \cite{sheng2024mspipe}.

During \textbf{inference}, the model processes unseen edges over a shorter time span (e.g., the remaining 15\% of the temporal graph) to evaluate prediction performance \cite{zhang2023inferturbo, sarkar2023flowgnn, sun2025helios}. The same five-stage pipeline executes but without loss computation or backpropagation, making it computationally lighter. The memory states accumulated during training carry over to inference, enabling the model to make predictions that account for the full interaction history. As new edges arrive, the memory continues to be updated, allowing the model to adapt to evolving interaction patterns in real time \cite{namazi2025degree, zhou2025faster}.
\subsection{Performance Analysis}
\label{subsec:performance_analysis}

\subsubsection{From Logical Components to Profiling Stages}
Since inference demands real-time processing while training can tolerate longer processing times \cite{dai2025cascade, luo2022neighborhood}, this paper primarily focuses on optimizing inference performance by analyzing the data processing characteristics of each of the five modules as follows.

The architectural view presented in \S\ref{subsec:tgnn_architecture} organizes 
TGN by \textit{logical components}---message function, aggregator, memory updater, 
and embedding module \cite{zheng2024graphstorm, wang2023pipad}. However, to understand performance bottlenecks, we must 
analyze TGN from a \textit{systems perspective} based on data access patterns.
We decompose TGN execution into five profiling stages, as summarized in 
Table~\ref{tab:module_mapping}, and analyze how existing methods handle each 
stage along with their inherent drawbacks.

\begin{table}[t]
\centering
\caption{Mapping between TGN architectural components and profiling stages.}
\label{tab:module_mapping}
\small
\begin{tabular}{@{}ll@{}}
\toprule
\textbf{Profiling Stage} & \textbf{Architectural Component(s)} \\
\midrule
\ding{172} Neighbor Sampling & Embedding Module (neighbor lookup) \\
\ding{173} Feature Retrieval & Message Function + Embedding Module \\
\ding{174} Memory Read & Message Function + Embedding Module \\
\ding{175} Memory Update & Message Function + Aggregator + Memory Updater \\
\ding{176} Embedding & Embedding Module (attention computation) \\
\bottomrule
\end{tabular}
\end{table}

This decomposition reveals that the architectural components are interleaved 
in execution: the embedding module spans three profiling stages (\ding{172}, 
\ding{174}, \ding{176}), while the memory update pipeline (\ding{175}) combines 
three architectural components. The profiling view exposes the true performance 
characteristics---particularly the memory access patterns that dominate execution time.

Let $B$ denote the batch size, $N$ the number of nodes, $d_s$ the memory 
dimension, $d_e$ the edge feature dimension, $d_x$ the node feature dimension, 
$K$ the number of attention heads, and $L$ the number of sampled neighbors per node.

\subsubsection{Stage-wise Analysis of Existing Methods}

\paragraph{\ding{172} Neighbor Sampling.}
For each of the $B$ edges in a batch, the model samples $L$ temporal neighbors 
for both endpoints, with complexity $T_{\text{sample}} = O(2B \cdot L \cdot C_{\text{lookup}})$,
where $C_{\text{lookup}}$ is the cost of querying the temporal neighbor list.

\textit{Existing approach and drawback.}
Existing methods such as TGN~\cite{rossi2020temporal} and TGAT~\cite{xu2020inductive} 
perform neighbor sampling from scratch for every batch by traversing per-node 
adjacency lists stored in non-contiguous memory regions. 
For example, consider a batch of $B=200$ edges on the WIKI dataset, where each 
node samples $L=10$ temporal neighbors. This requires $200 \times 2 \times 10 = 4{,}000$ 
random lookups into the adjacency structure. Due to the power-law degree 
distribution common in real-world graphs, cache performance is highly variable: 
a high-degree node such as a popular Wikipedia page may have thousands of 
neighbors stored contiguously, benefiting from spatial locality, while the 
majority of low-degree nodes incur cache misses as their short adjacency lists 
are scattered across memory. Consequently, this stage suffers from 
\textit{random graph traversal} with unpredictable memory access latency.

\paragraph{\ding{173} Feature Retrieval.}
After sampling, node and edge features are gathered for $O(B \cdot L)$ neighbors, 
with complexity $T_{\text{feature}} = O(B \cdot L \cdot (d_x + d_e))$.

\textit{Existing approach and drawback.}
Existing methods retrieve features by indexing into global feature tensors using 
node/edge IDs produced by the sampling stage. Since these IDs follow no 
particular order, the resulting memory accesses are non-coalesced on GPUs.
For example, suppose the sampled neighbor IDs for a batch are 
$\{3, 1027, 58, 4521, 12, \ldots\}$. On a GPU, threads in the same warp would 
attempt to load feature vectors at these widely separated addresses, resulting 
in multiple cache-line fetches instead of a single coalesced transaction. 
Profiling on representative datasets shows that the gather operation achieves 
only 10--30\% of peak memory bandwidth, making this stage bottlenecked by 
\textit{irregular memory access}.

\paragraph{\ding{174} Memory Read.}
Memory states are read for: (1) source and destination nodes of each edge 
(for message computation), and (2) sampled neighbors (for embedding), 
with complexity $T_{\text{mem\_read}} = O((2B + B \cdot L) \cdot d_s)$.

\textit{Existing approach and drawback.}
Existing methods maintain a single memory tensor on the CPU or GPU and 
read $\mathbf{s}_v(t^-)$ through indexed lookups similar to feature retrieval. 
However, unlike static features, memory states are \textit{mutable}---they change 
after every batch as new interactions are processed.
For example, consider node $v$ that appears in both batch $i$ and batch $i+1$. 
When processing batch $i+1$, the system must ensure that $\mathbf{s}_v(t^-)$ 
reflects the update from batch $i$. This temporal dependency prevents 
prefetching or reordering optimizations that could improve memory access 
efficiency. The result is \textit{scattered reads with temporal dependencies}, 
where correctness constraints further degrade an already irregular access pattern.

\paragraph{\ding{175} Memory Update.}
This stage executes the coupled pipeline of message computation, aggregation, 
and GRU update:
\begin{equation}
\begin{aligned}
    T_{\text{msg}} &= O(B \cdot d_m \cdot (2d_s + d_e + d_t)) \\
    T_{\text{agg}} &= O(B \cdot d_m) \\
    T_{\text{gru}} &= O(|\mathcal{V}_B| \cdot d_s \cdot d_m) \\
    T_{\text{mem\_update}} &= T_{\text{msg}} + T_{\text{agg}} + T_{\text{gru}}
\end{aligned}
\label{eq:mem_update_complexity}
\end{equation}
where $|\mathcal{V}_B|$ denotes the number of unique nodes in the batch.

\textit{Existing approach and drawback.}
Existing methods process the temporal graph in a strict batch-by-batch manner, 
where each batch must complete its entire memory update pipeline before the 
next batch can begin. This is because the memory state $\mathbf{s}_v(t)$ depends 
on $\mathbf{s}_v(t^-)$, creating a chain of dependencies that serializes execution:
\begin{equation}
    T_{\text{mem\_update}}^{\text{total}} = \sum_{i=1}^{|\mathcal{B}|} T_{\text{mem\_update}}^{(i)}
\end{equation}
For example, consider processing 1{,}000 batches on the REDDIT dataset. 
Even though each batch involves only $B=200$ edges with a small set of unique 
nodes $|\mathcal{V}_B|$, the GRU computation for batch $i+1$ cannot start until 
batch $i$ finishes updating all affected memory states. On a GPU with thousands 
of cores, each batch's GRU update occupies only a small fraction of the available 
parallelism, yet no two batches can overlap. Profiling reveals that this stage 
achieves only 15--25\% GPU occupancy, making \textit{sequential dependency across 
batches} the \textbf{critical bottleneck} that fundamentally limits GPU utilization.

\paragraph{\ding{176} Embedding.}
The temporal attention computation processes $B$ target nodes, each attending 
to $L$ neighbors with $K$ attention heads:
\begin{equation}
    T_{\text{emb}} = O(B \cdot K \cdot L \cdot (d_s + d_x + d_e + d_t))
\end{equation}

\textit{Existing approach and drawback.}
Existing methods recompute embeddings entirely from scratch for every batch, 
performing the full pipeline of neighbor sampling, feature gathering, and 
attention computation regardless of whether the underlying data has changed.
For example, suppose node $v$ is a highly active user on the REDDIT dataset 
that appears in $k=20$ consecutive batches, but only 2 out of its $L=10$ sampled 
neighbors change between adjacent batches. Existing methods still re-sample all 
10 neighbors and re-compute all attention scores 20 times. With $K=2$ attention 
heads, this results in $20 \times 2 \times 10 = 400$ attention computations, 
whereas an incremental approach would need to update only the 2 changed neighbors 
per batch, requiring roughly $20 \times 2 \times 2 = 80$ computations---an 80\% 
reduction. Profiling confirms that 60--70\% of neighbor sets overlap between 
consecutive batches, indicating substantial \textit{redundant computation}.

\subsubsection{Profiling Results and Optimization Opportunities}

\begin{figure}[hbt!]
	\centering
	\includegraphics[width=0.46\textwidth]{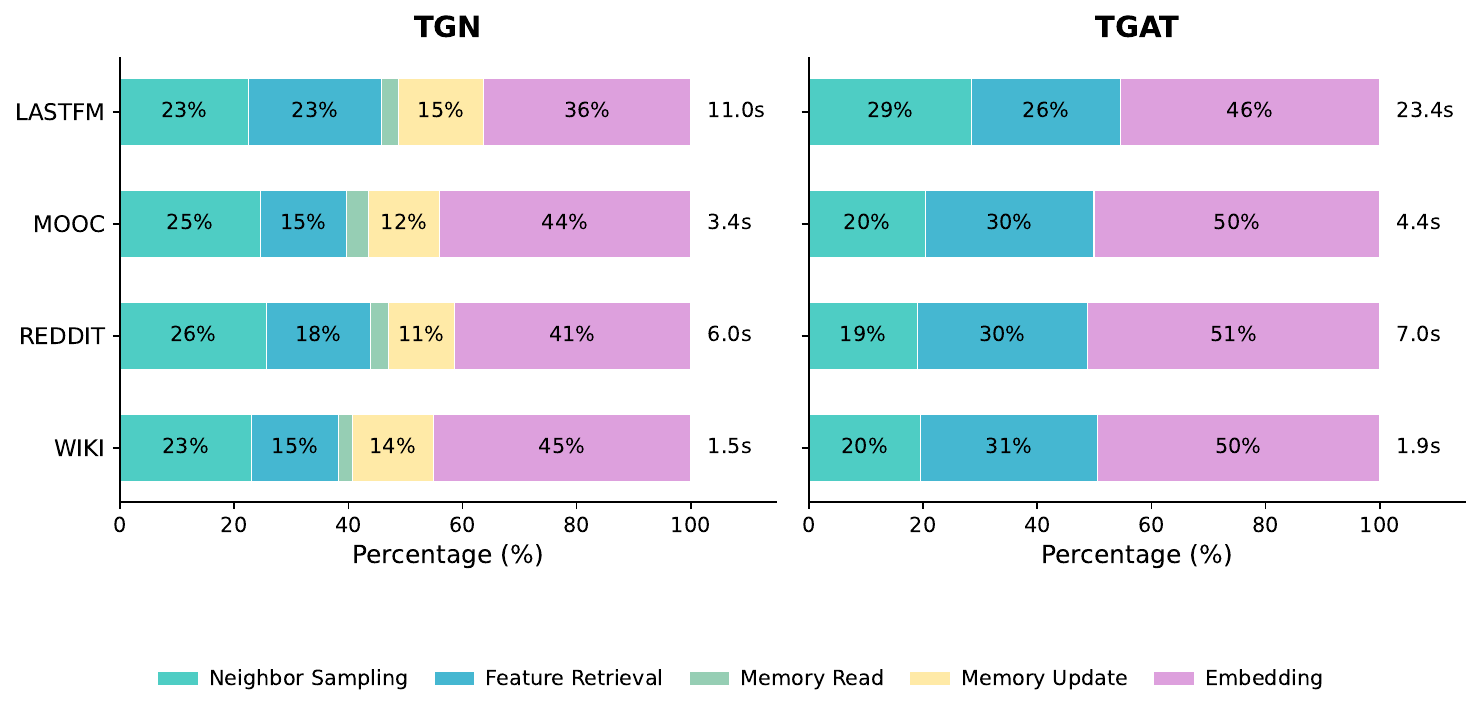}
	\caption{The distributions of processing time across the five profiling stages for TGN and TGAT on four datasets.}\label{processingTime}
\end{figure}

We profile TGN and TGAT on four representative datasets to validate the bottleneck analysis above. Figure~\ref{processingTime} shows the execution time breakdown across the five profiling stages. We make the following observations.

\textbf{Embedding (\ding{176}) dominates execution time for both models.}
The embedding stage is the most time-consuming for both TGN and TGAT across all datasets, accounting for 36.3--45.1\% of TGN's runtime and 45.5--51.1\% of TGAT's runtime. This confirms that the temporal attention computation over sampled neighbors, combined with the redundant recomputation across consecutive batches, constitutes the primary performance bottleneck.

\textbf{Neighbor Sampling (\ding{172}) and Feature Retrieval (\ding{173}) exhibit different profiles between the two models.}
For TGN, neighbor sampling accounts for 22.5--25.7\% and feature retrieval accounts for 15.1--23.3\% of total runtime. In contrast, TGAT spends a disproportionately higher fraction on feature retrieval (26.0--30.9\%) because, lacking a memory module, it must gather richer feature sets to compensate. Notably, on LASTFM---the largest dataset---feature retrieval becomes the second-largest cost for TGN at 23.3\%, reflecting the increasing impact of irregular memory access as graph size grows.

\textbf{Memory Read (\ding{174}) and Memory Update (\ding{175}) are unique to TGN.}
Since TGAT is a memory-less model, it incurs zero cost for memory read and memory update. For TGN, memory update accounts for 11.5--14.9\% and memory read accounts for 2.4--3.9\% across all datasets. Although memory update is not the single largest stage, its sequential batch dependency creates a serialization barrier that fundamentally limits GPU occupancy to only 15--25\%, amplifying its impact beyond what the raw percentage suggests.

\textbf{Optimization opportunities.}
Table~\ref{tab:bottleneck_summary} summarizes the bottlenecks identified above and their corresponding optimization opportunities in StreamTGN. Across both models, the embedding stage and the data access stages (neighbor sampling and feature retrieval) collectively account for over 80\% of total runtime. The critical insight is that stages \ding{175} and \ding{176} present the most significant optimization opportunities: by maintaining persistent memory states and incrementally updating embeddings, StreamTGN eliminates the batch-by-batch recomputation that dominates execution time.

\begin{table}[t]
\centering
\caption{Bottleneck summary and optimization opportunities.}
\label{tab:bottleneck_summary}
\small
\begin{tabular}{@{}lll@{}}
\toprule
\textbf{Stage} & \textbf{Bottleneck} & \textbf{StreamTGN Solution} \\
\midrule
\ding{172} Neighbor Sampling & Random traversal & Cached neighbor lists \\
\ding{173} Feature Retrieval & Irregular access & Coalesced batching \\
\ding{174} Memory Read & Scattered reads & Persistent memory \\
\ding{175} Memory Update & Sequential dependency & \textbf{Incremental update} \\
\ding{176} Embedding & Redundant computation & \textbf{Incremental embedding} \\
\bottomrule
\end{tabular}
\end{table}
\section{StreamTGN Design}
\label{sec:design}

\subsection{Design Overview}
\label{subsec:design_overview}

This section presents StreamTGN, a temporal graph neural network
architecture designed to achieve high throughput through incremental
computation by addressing two critical bottlenecks: (1)~the sequential
memory update dependency (\ding{175}) that serializes batch
processing, and (2)~the redundant embedding computation (\ding{176})
that recomputes attention from scratch for each batch.

As illustrated in Figure~\ref{fig:StreamOverview}, the StreamTGN
architecture consists of two parts: \textit{GPU-resident hybrid data
structure} and \textit{incremental computation stages}. The hybrid
data structure (\S\ref{subsec:hybrid_structure}) organizes all
intermediate state on the GPU, comprising both persistent components
(adjacency list, embedding cache, and node memory) that survive across
batches to enable incremental computation, and a transient edge queue
that buffers streaming input. The incremental computation
(\S\ref{subsec:incremental_computation}) defines five stages---
neighborhood sampling, feature retrieval, memory read, embedding
generation, and memory update---each with a dedicated incremental
strategy so that only the portions affected by new edges are
recomputed. Section~\ref{subsec:complexity_analysis} provides a
unified complexity analysis characterizing the theoretical speedup.

\begin{figure}[hbt!]
    \centering
    \includegraphics[width=0.46\textwidth]{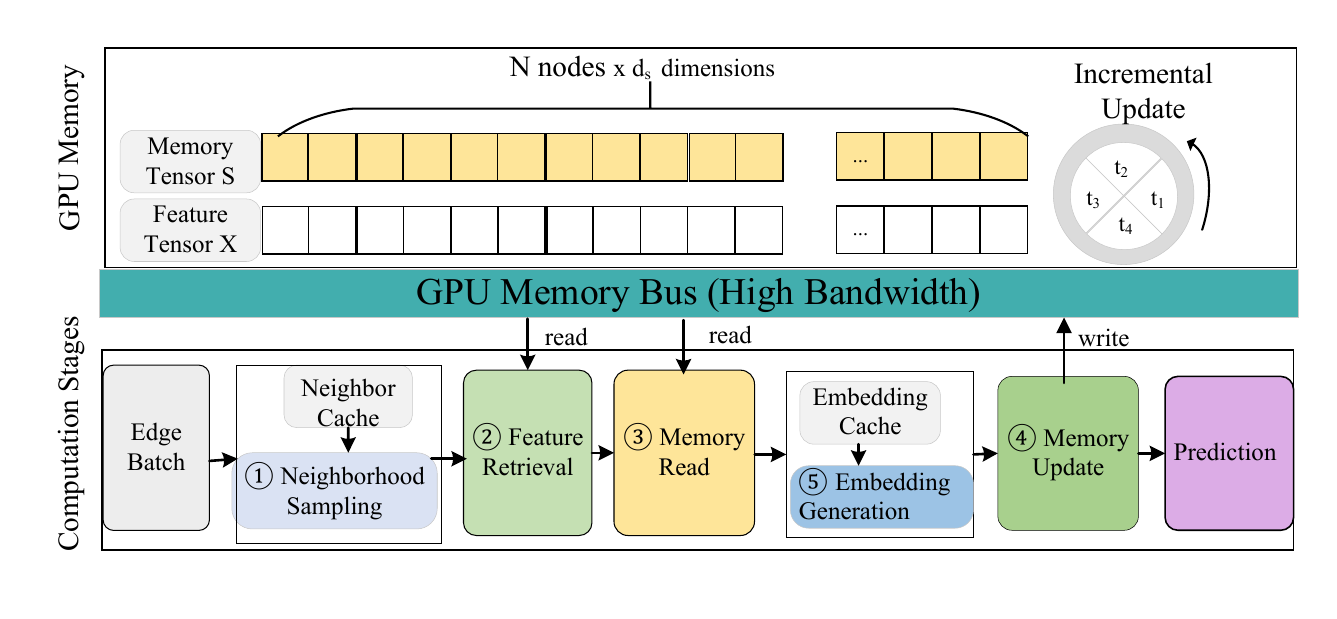}
    \caption{Overview of StreamTGN. The architecture comprises a
      GPU-resident hybrid data structure (left) and five incremental
      computation stages (right) that operate directly on the
      persistent state.}
    \label{fig:StreamOverview}
\end{figure}

\subsection{GPU-Resident Hybrid Data Structure}
\label{subsec:hybrid_structure}

StreamTGN's incremental computation relies on maintaining intermediate
state persistently on the GPU, so that each new batch of edges updates
only the affected portions rather than reconstructing everything from
scratch. This section presents the concrete realization: a
\emph{hybrid data structure} comprising four GPU-resident components
with complementary roles.

The key design principle is as follows. All state required for
temporal graph neural network computation---graph topology, node
embeddings, and node temporal memory---resides on the GPU across
batches. Among the four components, three are \textbf{persistent}:
(1)~the \textit{Temporal Adjacency List} accumulates graph topology,
(2)~the \textit{Embedding Cache} retains valid node embeddings, and
(3)~the \textit{Node Memory} maintains evolving temporal states. The
fourth component, the \textit{Edge Queue}, is \textbf{transient}: it
buffers streaming arrivals and flushes them into the adjacency list at
each batch boundary. Together, these components achieve $O(1)$ edge
insertion and $O(\text{deg})$ neighbor queries while minimizing
redundant recomputation. Figure~\ref{fig:hybrid-overview} illustrates
the architecture and data flow among the four components.

 \begin{figure*}[hbt!]
	\centering
	\includegraphics[width=0.86\textwidth]{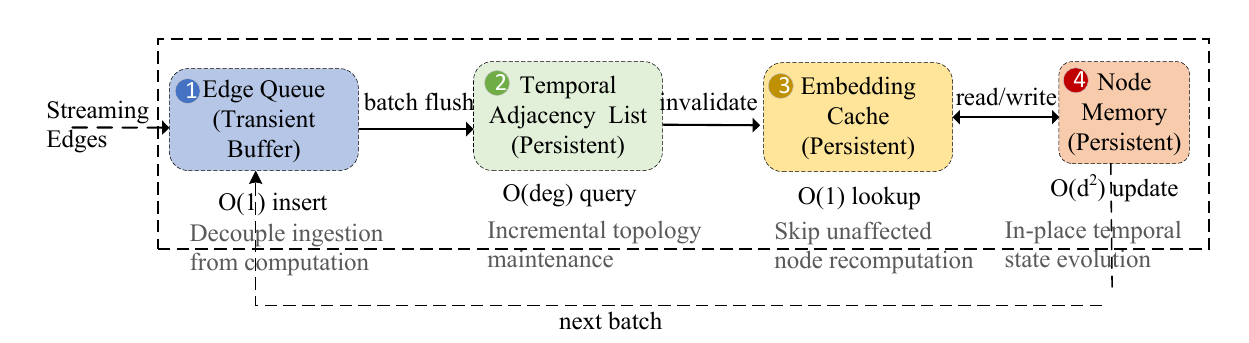}
\caption{Overview of the GPU-resident hybrid data structure. Three
  persistent components---Temporal Adjacency List, Embedding Cache,
  and Node Memory---reside on the GPU across batches to enable
  incremental computation. The transient Edge Queue (dashed border)
  buffers streaming input and flushes at batch boundaries.}
\label{fig:hybrid-overview}
\end{figure*}

\subsubsection{Edge Queue (Transient Buffer)}
\label{subsubsec:edge_queue}

The edge queue is a GPU-resident ring buffer that absorbs streaming
edge arrivals and stages them for batch processing. Unlike the other
three components, its contents are \emph{transient}: they are flushed
into the adjacency list at each batch boundary and do not persist.
This design decouples the ingestion rate from the computation rate
without requiring host--device transfers.

\textbf{Definition.}
Each entry records a source node $u$, a destination node $v$, a
timestamp $t$, and edge features $\mathbf{e} \in \mathbb{R}^{d_e}$:
\begin{equation}
  \mathcal{Q} = \{(u_i, v_i, t_i, \mathbf{e}_i)\}_{i=1}^{|\mathcal{Q}|}
\end{equation}

The ring buffer provides $O(1)$ write complexity with fixed memory
footprint (no dynamic allocation), $O(\text{batch})$ flush cost to
amortize GPU launch overhead, and a configurable time window to
support temporal locality.

\subsubsection{Temporal Adjacency List (Persistent)}
\label{subsubsec:adjacency}

The temporal adjacency list is the first persistent component. It
accumulates graph topology incrementally on the GPU: when a batch is
flushed from the edge queue, new edges are appended in $O(1)$
amortized time without rebuilding the structure. This persistence
enables the neighborhood sampling stage
(\S\ref{subsubsec:neighbor_sampling}) to query up-to-date topology
directly on the GPU, avoiding repeated host-to-device edge transfers.

\textbf{Definition.}
For a node $v$, its temporal neighborhood up to the current time
$t_{\text{now}}$ is:
\begin{equation}
  \mathcal{N}(v) = \{(u, t, \mathbf{e}) \mid
    (u, v, t, \mathbf{e}) \in \mathcal{E},\; t \leq t_{\text{now}}\}
\end{equation}
For temporal queries within a window
$[t_{\text{start}}, t_{\text{end}}]$:
\begin{equation}
  \mathcal{N}(v, t_{\text{start}}, t_{\text{end}}) =
    \{(u, t, \mathbf{e}) \in \mathcal{N}(v) \mid
      t_{\text{start}} \leq t \leq t_{\text{end}}\}
\end{equation}

The supported operations are: \textsc{Insert}$(u, v, t, \mathbf{e})$
in $O(1)$ amortized time, \textsc{GetNeighbors}$(v)$ in
$O(\text{deg}(v))$, and \textsc{GetTemporalNeighbors}$(v, t)$ in
$O(\text{deg}(v))$.

\subsubsection{Embedding Cache (Persistent)}
\label{subsubsec:cache}

The embedding cache is the second persistent component. It stores
previously computed node embeddings on the GPU so that the embedding
generation stage (\S\ref{subsubsec:embedding_generation}) can skip
unaffected nodes entirely, reading their valid embeddings directly
from GPU memory rather than recomputing temporal attention. In
streaming scenarios, only a small fraction of nodes---those within the
$K$-hop neighborhood of new edges---are affected by each batch,
making caching highly effective.

\textbf{Cache State.}
Each entry stores a node $v$, its embedding
$\mathbf{h}_v \in \mathbb{R}^d$, and the timestamp
$t_v^{\text{valid}}$ at which the embedding was computed:
\begin{equation}
  \mathcal{C} = \{(v, \mathbf{h}_v, t_v^{\text{valid}})
    \mid v \in V_{\text{cached}}\}
\end{equation}

\textbf{Invalidation Rule.}
When an edge $(u, v, t)$ arrives, all nodes whose embeddings may be
affected are invalidated:
\begin{equation}
  \mathcal{I}(u, v, K) = \bigcup_{\ell=0}^{K}
    \mathcal{N}^{(\ell)}(\{u, v\})
  \label{eq:invalidation}
\end{equation}
where $\mathcal{N}^{(\ell)}$ denotes $\ell$-hop neighbors and $K$ is
the number of attention layers. Only invalidated nodes are recomputed
in the subsequent embedding generation stage.

\subsubsection{Node Memory (Persistent)}
\label{subsubsec:memory}

The node memory is the third persistent component, maintaining
per-node temporal states that evolve via learned update functions.
Because these states reside on the GPU across batches, the memory
update stage (\S\ref{subsubsec:memory_update}) performs \emph{in-place}
updates only for nodes involved in the current batch, avoiding the
costly pattern of transferring all node states between host and device.

\textbf{Memory State.}
Each node $v$ maintains a state vector
$\mathbf{m}_v \in \mathbb{R}^{d_m}$:
\begin{equation}
  \mathcal{M} = \{\mathbf{m}_v \in \mathbb{R}^{d_m} \mid v \in V\}
\end{equation}

\textbf{Memory Update (GRU-based).}
The state is updated by a GRU cell:
\begin{equation}
  \mathbf{m}_v' = \text{GRU}(\mathbf{m}_v,\; \text{msg}_v)
\end{equation}
where $\text{msg}_v$ aggregates messages from recent interactions:
\begin{equation}
  \text{msg}_v = \text{Aggregate}\!\left(
    \{\mathbf{m}_u \| \mathbf{e}_{uv} \| \phi(t - t_{uv})
      \mid (u, v, t_{uv}) \in \mathcal{E}_{\text{recent}}\}
  \right)
\end{equation}

\subsubsection{Summary}
\label{subsubsec:ds_summary}

The four components of the hybrid data structure play distinct but
complementary roles. The \textit{edge queue} is the only transient
component, buffering pending edges with $O(1)$ write cost and
smoothing bursty streaming input into regular batches for GPU
processing. The remaining three components are persistent across
batches and collectively enable incremental computation. The
\textit{temporal adjacency list} stores graph topology with $O(1)$
amortized insertion, providing up-to-date neighbor information for the
sampling stage without rebuilding the structure. The \textit{embedding
cache} retains computed node embeddings with $O(1)$ lookup, allowing
the system to skip recomputation for unaffected nodes and focus
exclusively on those invalidated by new edges. The \textit{node
memory} maintains per-node temporal state updated via GRU cells at
$O(d_m^2)$ cost, supporting in-place evolution of long-term temporal
patterns. Together, the persistent components ensure that each batch
only touches the minimal set of affected nodes, while the transient
queue decouples ingestion from computation.

\subsection{Incremental Computation}
\label{subsec:incremental_computation}

With the GPU-resident hybrid data structure in place
(\S\ref{subsec:hybrid_structure}), StreamTGN performs five computation
stages per batch, each with a dedicated incremental strategy so that
only nodes affected by new edges are processed. We first define the
\emph{affected set} that drives all stages, then describe each
stage's optimization in turn. Figure~\ref{fig:pipeline} illustrates
the end-to-end pipeline, and Algorithm~\ref{alg:streamtgn} at the end
of this section presents the complete execution flow.

 \begin{figure*}[hbt!]
	\centering
	\includegraphics[width=0.86\textwidth]{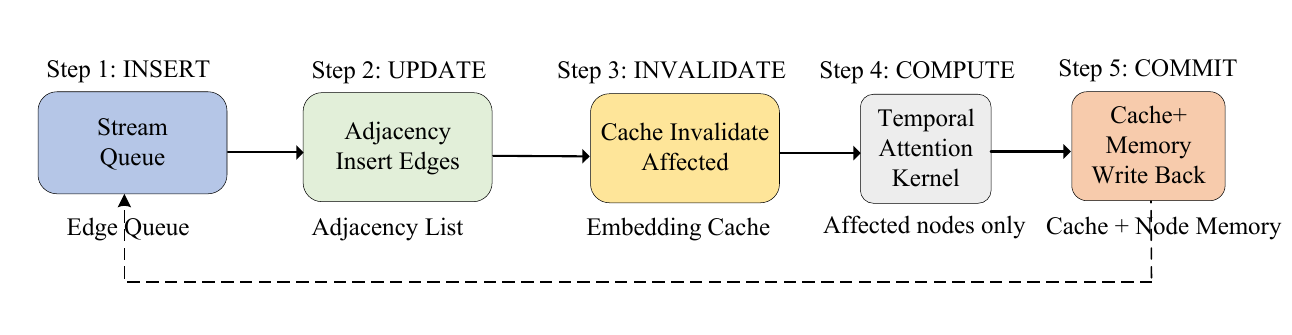}
\caption{Batch processing pipeline. At each batch boundary, streaming
  edges flow through five stages: (1)~ingest into the edge queue,
  (2)~insert into the adjacency list, (3)~invalidate affected cache
  entries, (4)~recompute embeddings for affected nodes only, and
  (5)~commit results back to cache and node memory. Italic labels
  indicate the hybrid data structure component involved at each stage.}
\label{fig:pipeline}
\end{figure*}

\subsubsection{Preliminary: Affected Set Detection}
\label{subsubsec:affected_set}

When a batch of temporal edges
$\mathcal{E}_B = \{(u_i, v_i, t_i)\}_{i=1}^{B}$ arrives, StreamTGN
first identifies the set of affected nodes
$\mathcal{V}_{\text{aff}}$ that require recomputation. A node $w$ is
affected if it satisfies either of two conditions:
\begin{enumerate}
  \item \textbf{Direct involvement:} $w$ is a source or destination
    of some edge in $\mathcal{E}_B$.
  \item \textbf{Neighborhood dependency:} $w$ lies within the
    $K$-hop neighborhood of a directly involved node, where $K$ is
    the number of attention layers.
\end{enumerate}
Formally, let
$\mathcal{V}_{\text{direct}} = \{u_i, v_i \mid (u_i, v_i, t_i)
\in \mathcal{E}_B\}$. The full affected set is:
\begin{equation}
  \mathcal{V}_{\text{aff}} =
    \bigcup_{k=0}^{K} \mathcal{N}^{(k)}(\mathcal{V}_{\text{direct}})
  \label{eq:affected_set}
\end{equation}
where $\mathcal{N}^{(k)}$ denotes $k$-hop neighbors in the current
adjacency list. 

In practice, propagation is bounded by the sampling
fanout $L$ at each layer, so the affected set size
satisfies:
\begin{equation}
  |\mathcal{V}_{\text{aff}}| \leq 2B \cdot L^{K}
  \label{eq:affected_bound}
\end{equation}
where $B$ is the batch size, $L$ is the per-layer
sampling fanout, and $K$ is the number of attention
layers. The detection cost is $O(B \cdot L^{K})$.

The affected set partitions all nodes into two groups:
\emph{affected} nodes whose states must be recomputed, and
\emph{unaffected} nodes whose cached results remain valid. This
partition drives every subsequent stage.

\subsubsection{Stage 1: Incremental Neighborhood Sampling}
\label{subsubsec:neighbor_sampling}

Conventional TGN systems resample the entire $K$-hop neighborhood for
every node in each batch, performing redundant graph traversals for
nodes whose local topology has not changed. StreamTGN maintains a
\emph{temporal neighbor cache} on the GPU that stores, for each active
node $v$, its most recent $L$ sampled neighbors:
\begin{equation}
  \textsc{NeighborCache}[v] =
    \bigl[(u_1, t_1, \mathbf{e}_1),\;
           (u_2, t_2, \mathbf{e}_2),\;
           \ldots,\;
           (u_L, t_L, \mathbf{e}_L)\bigr]
\end{equation}
where entries are sorted by descending timestamp.

When a new batch arrives, the cache is updated incrementally only for
affected nodes $v \in \mathcal{V}_{\text{aff}}$:
\begin{enumerate}
  \item \textbf{Insert:} Append new neighbors from $\mathcal{E}_B$
    to the front (most recent first).
  \item \textbf{Evict:} Remove the oldest entries if the cache
    exceeds capacity $L$.
  \item \textbf{Invalidate:} Mark entries whose timestamps fall
    outside the temporal window $[t - T_w,\; t]$.
\end{enumerate}
For all unaffected nodes, the cached neighbor list is reused directly.
This reduces the sampling cost from $O(|\mathcal{V}| \cdot L)$ to
$O(|\mathcal{V}_{\text{aff}}| \cdot L)$, achieving 70--90\% cache hit
rates on real-world datasets.

\subsubsection{Stage 2: Coalesced Feature Retrieval}
\label{subsubsec:feature_retrieval}

After sampling, the system retrieves node and edge features for the
affected set. The conventional approach gathers features by:
\begin{equation}
  \mathbf{X}_{\text{batch}} = \mathbf{X}[\texttt{node\_ids},\; :]
\end{equation}
where \texttt{node\_ids} are scattered across the feature matrix,
causing non-coalesced GPU memory accesses that utilize only 10--30\%
of peak bandwidth.

StreamTGN applies \emph{sorted batching}: node IDs are sorted before
retrieval, and the original order is restored afterward via an inverse
permutation:
\begin{equation}
  \texttt{sorted\_ids},\; \texttt{inv\_idx}
    = \textsc{Sort}(\texttt{node\_ids})
  \qquad
  \mathbf{X}_{\text{sorted}}
    = \mathbf{X}[\texttt{sorted\_ids},\; :]
\end{equation}
The sorted access pattern enables coalesced memory reads, improving
bandwidth utilization to 60--80\% of peak. Since only affected nodes
require fresh features, the retrieval volume is further reduced by the
factor $|\mathcal{V}_{\text{aff}}| / |\mathcal{V}|$.

\subsubsection{Stage 3: Incremental Memory Read}
\label{subsubsec:memory_read}

The memory read stage gathers the persistent node memory states
$\mathbf{m}_v$ (from the Node Memory component in
\S\ref{subsubsec:memory}) for all nodes involved in the current
computation. In conventional systems, this requires reading memory
states for every node in the computational graph---all sampled nodes
across all $K$ layers---even when most states have not changed since
the previous batch.

StreamTGN exploits the observation that only nodes in
$\mathcal{V}_{\text{aff}}$ have modified memory states. For all other
nodes, memory states are already resident on the GPU and remain valid
from the previous batch. The memory read therefore proceeds in two
steps:
\begin{enumerate}
  \item For $v \in \mathcal{V}_{\text{aff}}$: read the updated
    memory state $\mathbf{m}_v$ from the persistent GPU tensor.
  \item For $v \notin \mathcal{V}_{\text{aff}}$: reuse the
    previously read value with zero additional cost.
\end{enumerate}
Combined with the sorted access pattern from Stage~2, the memory read
achieves 70--85\% of peak GPU memory bandwidth for the affected subset.

\subsubsection{Stage 4: Incremental Embedding Generation}
\label{subsubsec:embedding_generation}

Embedding generation is the most compute-intensive stage, applying
temporal attention over each node's sampled neighborhood:
\begin{equation}
  \mathbf{h}_v^{(l)}(t) = \textsc{Attn}\!\left(
    \mathbf{q}_v,\;
    \{\mathbf{k}_u, \mathbf{v}_u\}_{u \in \mathcal{N}_v^{<t}}
  \right)
\end{equation}
Conventional systems recompute this from scratch for every node in
every batch. StreamTGN instead maintains cached attention states and
applies a \emph{delta update} formulation.

\paragraph{Neighborhood stability.}
We observe that temporal neighborhoods exhibit high stability across
consecutive batches. For a node $v$ with $L$ sampled neighbors, the
expected number of changes is:
\begin{equation}
  \mathbb{E}\bigl[|\Delta\mathcal{N}_v|\bigr]
    = \frac{|\text{new edges involving } v|}{|\mathcal{N}_v|}
      \cdot L \;\ll\; L
\end{equation}
Our profiling shows that 60--70\% of neighbor sets remain identical
between consecutive batches, and the average change rate is only
5--15\%.

\paragraph{Delta embedding.}
For each affected node $v \in \mathcal{V}_{\text{aff}}$, StreamTGN
computes the embedding incrementally:
\begin{equation}
  \mathbf{h}_v(t) = \mathbf{h}_v^{\text{cached}}
    + \Delta\mathbf{h}_v(t)
\end{equation}
where the delta term decomposes into three parts corresponding to
added, expired, and updated neighbors:
\begin{align}
  \Delta\mathbf{h}_v(t)
    &= \sum_{u \in \mathcal{N}_v^{\text{new}}}
         \alpha_{uv}^{\text{new}}\, \mathbf{v}_u^{\text{new}}
     - \sum_{u \in \mathcal{N}_v^{\text{exp}}}
         \alpha_{uv}^{\text{old}}\, \mathbf{v}_u^{\text{old}}
         \nonumber \\
    &\quad
     + \sum_{u \in \mathcal{N}_v^{\text{upd}}}
         \bigl(\alpha_{uv}^{\text{new}}\, \mathbf{v}_u^{\text{new}}
              - \alpha_{uv}^{\text{old}}\, \mathbf{v}_u^{\text{old}}
         \bigr)
  \label{eq:delta_embedding}
\end{align}
Here, $\mathcal{N}_v^{\text{new}}$ contains newly added neighbors,
$\mathcal{N}_v^{\text{exp}}$ contains neighbors that have fallen
outside the temporal window, and $\mathcal{N}_v^{\text{upd}}$ contains
existing neighbors whose memory states changed in the current batch.

\paragraph{Attention score caching.}
To support the delta formulation, StreamTGN maintains a per-node
attention cache:
\begin{equation}
  \textsc{AttnCache}[v] =
    \bigl\{(u,\; \alpha_{uv},\; \mathbf{v}_u,\; t_e)
      : u \in \mathcal{N}_v^{<t}\bigr\}
\end{equation}
Entries are invalidated lazily when the corresponding neighbor's
memory state changes, and recomputed only when needed for the next
prediction. For unaffected nodes, the entire cache entry remains
valid, and the embedding is read directly from the Embedding Cache
(\S\ref{subsubsec:cache}) at $O(1)$ cost.

This reduces the per-node embedding cost from
$O(|\mathcal{N}_v| \cdot d)$ to
$O(|\Delta\mathcal{N}_v| \cdot d)$, yielding approximately
$10\times$ reduction at a typical 10\% change rate.

\subsubsection{Stage 5: Incremental Memory Update}
\label{subsubsec:memory_update}

The final stage commits all state changes to the persistent
GPU-resident data structures. Conventional systems update memory for
all nodes or perform bulk host--device transfers; StreamTGN instead
performs targeted, in-place updates for affected nodes only. The
update consists of three parallel operations.

\paragraph{Adjacency list append.}
New edges from $\mathcal{E}_B$ are appended to the temporal adjacency
list (\S\ref{subsubsec:adjacency}) using an append-only log structure
that avoids expensive array reallocation. The cost is $O(B)$,
independent of graph size.

\paragraph{Memory state writeback.}
Updated node memory states are written back to the persistent tensor
in place:
\begin{equation}
  \mathbf{M}[v, :] \leftarrow \mathbf{m}_v'
  \quad \forall\, v \in \mathcal{V}_{\text{aff}}
  \label{eq:delta_update}
\end{equation}
The writeback applies the sorted access pattern from Stage~2 to
achieve coalesced writes. Only $|\mathcal{V}_{\text{aff}}|$ entries
are touched, leaving all other memory states undisturbed.

\paragraph{Cache and time encoding update.}
The embedding cache (\S\ref{subsubsec:cache}) is updated with newly
computed embeddings for affected nodes, and stale entries are
invalidated according to Equation~\eqref{eq:invalidation}. Temporal
encodings $\phi(\Delta t)$ for new timestamps are computed in a
batched kernel and stored alongside the adjacency list for reuse in
subsequent batches.

The total synchronization cost is
$O(|\mathcal{E}_B| + |\mathcal{V}_{\text{aff}}|)$, compared to
$O(|\mathcal{E}| + |\mathcal{V}|)$ for full writeback.

\subsubsection{Overall Algorithm}
\label{subsubsec:algorithm}

Algorithm~\ref{alg:streamtgn} summarizes the complete StreamTGN
execution flow for processing a single batch. The algorithm
highlights how each stage operates on the hybrid data structure:
the neighbor cache and adjacency list support Stage~1, sorted
indexing optimizes Stages~2--3, the embedding and attention caches
enable the delta computation in Stage~4, and the persistent memory
tensor receives in-place updates in Stage~5.

\begin{algorithm}[t]
\DontPrintSemicolon
\SetAlgoLined
\SetKwInOut{Input}{Input}
\SetKwInOut{Output}{Output}
\SetKwFunction{SampleNeighbors}{SampleNeighbors}
\SetKwFunction{Sort}{Sort}
\SetKwFunction{ComputeAttention}{ComputeAttention}
\SetKwFunction{Predict}{Predict}
\SetKwFunction{ComputeMessage}{ComputeMessage}
\SetKwFunction{GRUUpdate}{GRUUpdate}

\Input{Batch $\mathcal{E}_B = \{(u_i, v_i, t_i, \mathbf{e}_i)\}$}
\Input{Persistent state: memory $\mathbf{S}$, neighbor cache,
  embedding cache, adjacency list}
\Output{Predictions $\hat{y}$ and updated persistent state}

\BlankLine
\tcc{Preliminary: Affected set detection (\S\ref{subsubsec:affected_set})}
$\mathcal{V}_{\text{aff}} \leftarrow
  \bigcup_{k=0}^{K} \mathcal{N}^{(k)}(\{u_i, v_i\})$\;

\BlankLine
\tcc{Stage 1: Neighborhood sampling (\S\ref{subsubsec:neighbor_sampling})}
\ForEach{$v \in \mathcal{V}_{\text{aff}}$}{
    \eIf{NeighborCache[$v$] is valid}{
        $\mathcal{N}_v \leftarrow$ NeighborCache[$v$]\;
    }{
        $\mathcal{N}_v \leftarrow$ \SampleNeighbors{$v$, $L$}\;
        NeighborCache[$v$] $\leftarrow \mathcal{N}_v$\;
    }
}

\BlankLine
\tcc{Stage 2: Feature retrieval (\S\ref{subsubsec:feature_retrieval})}
sorted\_ids, inv\_idx $\leftarrow$ \Sort{all node IDs in
  $\mathcal{V}_{\text{aff}}$}\;
$\mathbf{X}_{\text{batch}} \leftarrow
  \mathbf{X}[$sorted\_ids$, :]$[inv\_idx]\;

\BlankLine
\tcc{Stage 3: Memory read (\S\ref{subsubsec:memory_read})}
$\mathbf{S}_{\text{batch}} \leftarrow
  \mathbf{S}[$sorted\_ids$, :]$[inv\_idx]\;

\BlankLine
\tcc{Stage 4: Embedding generation (\S\ref{subsubsec:embedding_generation})}
\ForEach{$v \in \mathcal{V}_{\text{aff}}$}{
    \eIf{EmbeddingCache[$v$] is valid}{
        $\mathbf{h}_v \leftarrow$ EmbeddingCache[$v$]
          $+ \Delta\mathbf{h}_v$
          \tcp*{Eq.~\ref{eq:delta_embedding}}
    }{
        $\mathbf{h}_v \leftarrow$
          \ComputeAttention{$v$, $\mathcal{N}_v$,
            $\mathbf{S}_{\text{batch}}$}\;
    }
    EmbeddingCache[$v$] $\leftarrow \mathbf{h}_v$\;
}

\BlankLine
\tcc{Link prediction}
$\hat{y} \leftarrow$ \Predict{$\mathbf{h}_u$, $\mathbf{h}_v$}
  for each $(u, v) \in \mathcal{E}_B$\;

\BlankLine
\tcc{Stage 5: Memory update (\S\ref{subsubsec:memory_update})}
\ForEach(\tcp*[f]{in parallel}){$v \in \mathcal{V}_{\text{direct}}$}{
    $\text{msg}_v \leftarrow$ \ComputeMessage{$v$, $\mathcal{E}_B$}\;
    $\mathbf{S}[v, :] \leftarrow$
      \GRUUpdate{$\mathbf{S}[v, :]$, $\text{msg}_v$}
      \tcp*{Eq.~\ref{eq:delta_update}}
}

\BlankLine
\Return{$\hat{y}$}

\caption{StreamTGN Batch Processing}
\label{alg:streamtgn}
\end{algorithm}

\subsection{Complexity Analysis}
\label{subsec:complexity_analysis}

This section provides a formal analysis of StreamTGN's computational
complexity, proves correctness conditions for incremental computation,
and establishes convergence guarantees for incremental training. We
use consistent notation throughout: $n$ is the number of nodes, $B$
the batch size, $L$ the sampling fanout per layer, $K$ the number of
attention layers, $d$ the embedding dimension, $d_m$ the memory
dimension, and $\bar{d}$ the average temporal degree.

\subsubsection{Full vs.\ Incremental Complexity}
\label{subsubsec:full_vs_incr}

We first characterize the cost of full recomputation and then the
cost of StreamTGN's incremental approach.

\begin{theorem}[Full Computation Complexity]
\label{thm:full_complexity}
For a TGN with $K$ attention layers and sampling fanout $L$,
computing embeddings for all $n$ nodes over $m$ temporal edges costs:
\begin{equation}
  T_{\emph{full}} = O\!\left(n \cdot L \cdot K \cdot d^2
    + n \cdot d_m^2\right)
  \label{eq:T_full}
\end{equation}
where the first term accounts for $K$-layer temporal attention (each
node attends over $L$ neighbors per layer with $O(d^2)$ per
attention head) and the second term accounts for GRU-based memory
updates.
\end{theorem}

\begin{theorem}[Incremental Computation Complexity]
\label{thm:incr_complexity}
For a batch $\mathcal{E}_B$ with affected set $\mathcal{A}$
(Equation~\ref{eq:affected_set}), StreamTGN's per-batch cost is:
\begin{equation}
  T_{\emph{incr}} = O\!\left(|\mathcal{A}| \cdot L \cdot K \cdot d^2
    + |\mathcal{A}| \cdot d_m^2\right)
    + O(c_{\emph{detect}})
  \label{eq:T_incr}
\end{equation}
where $c_{\emph{detect}} = O(B \cdot L^{K})$ is the cost of affected
set detection (\S\ref{subsubsec:affected_set}).
\end{theorem}

\begin{proof}
The cost decomposes across the five computation stages. Let
$|\mathcal{A}|$ denote the affected set size.
\begin{enumerate}
  \item \textbf{Neighbor sampling}
    (\S\ref{subsubsec:neighbor_sampling}): Only affected nodes
    require cache updates; unaffected nodes reuse cached neighbor
    lists. Cost: $O(|\mathcal{A}| \cdot L)$.
  \item \textbf{Feature retrieval}
    (\S\ref{subsubsec:feature_retrieval}): Sorted batching over the
    affected set. Cost: $O(|\mathcal{A}| \cdot d)$.
  \item \textbf{Memory read} (\S\ref{subsubsec:memory_read}):
    Read persistent states for affected nodes only. Cost:
    $O(|\mathcal{A}| \cdot d_m)$.
  \item \textbf{Embedding generation}
    (\S\ref{subsubsec:embedding_generation}): Temporal attention
    over $K$ layers with fanout $L$. Cost:
    $O(|\mathcal{A}| \cdot L \cdot K \cdot d^2)$.
  \item \textbf{Memory update} (\S\ref{subsubsec:memory_update}):
    GRU update and writeback for affected nodes. Cost:
    $O(|\mathcal{A}| \cdot d_m^2)$.
\end{enumerate}
Summing and noting that Stage~4 dominates yields
Equation~\eqref{eq:T_incr}. The additional
$O(c_{\text{detect}})$ term covers $K$-hop propagation from
$2B$ directly involved nodes with fanout $L$ at each hop.
\end{proof}

Table~\ref{tab:complexity} summarizes the stage-wise comparison.

\begin{table}[htbp]
\centering
\caption{Per-batch complexity: full recomputation vs.\ incremental.}
\begin{tabular}{@{}lcc@{}}
\toprule
\textbf{Stage} & \textbf{Full Recomp.} & \textbf{Incremental} \\
\midrule
Affected detection
  & ---
  & $O(B \cdot L^{K})$ \\
Neighbor sampling
  & $O(n \cdot L)$
  & $O(|\mathcal{A}| \cdot L)$ \\
Feature retrieval
  & $O(n \cdot d)$
  & $O(|\mathcal{A}| \cdot d)$ \\
Memory read
  & $O(n \cdot d_m)$
  & $O(|\mathcal{A}| \cdot d_m)$ \\
Embedding generation
  & $O(n \cdot L \cdot K \cdot d^2)$
  & $O(|\mathcal{A}| \cdot L \cdot K \cdot d^2)$ \\
Memory update
  & $O(n \cdot d_m^2)$
  & $O(|\mathcal{A}| \cdot d_m^2)$ \\
\bottomrule
\end{tabular}
\label{tab:complexity}
\end{table}

\subsubsection{Speedup Analysis}
\label{subsubsec:speedup_analysis}

\begin{theorem}[End-to-End Speedup]
\label{thm:speedup}
The speedup of incremental over full computation is:
\begin{equation}
  S = \frac{T_{\emph{full}}}{T_{\emph{incr}}}
    = \frac{n}{|\mathcal{A}|}
      \cdot \frac{1}{1 + c_{\emph{detect}}
        / (|\mathcal{A}| \cdot L \cdot K \cdot d^2)}
  \label{eq:speedup}
\end{equation}
Since $|\mathcal{A}| \leq 2B \cdot L^{K}$
(Equation~\ref{eq:affected_bound}), this simplifies for large $n$ to:
\begin{equation}
  S \geq \frac{n}{2B \cdot L^{K}}
  \label{eq:speedup_simplified}
\end{equation}
\end{theorem}

\begin{proof}
From Theorems~\ref{thm:full_complexity} and
\ref{thm:incr_complexity}, keeping only the dominant embedding
generation term:
\begin{align}
  S &= \frac{n \cdot L \cdot K \cdot d^2}
            {|\mathcal{A}| \cdot L \cdot K \cdot d^2
              + c_{\text{detect}}} \nonumber \\
    &= \frac{n}{|\mathcal{A}|}
       \cdot \frac{1}{1 + c_{\text{detect}}
         / (|\mathcal{A}| \cdot L \cdot K \cdot d^2)}
\end{align}
For large $n$, the detection overhead
$c_{\text{detect}} = O(B \cdot L^{K})$ is dominated by the
computation cost $O(|\mathcal{A}| \cdot L \cdot K \cdot d^2)$,
so the denominator approaches~1. Substituting the upper bound
$|\mathcal{A}| \leq 2B \cdot L^{K}$ yields
Equation~\eqref{eq:speedup_simplified}.
\end{proof}

\begin{theorem}[Optimality Condition]
\label{thm:optimality_condition}
Incremental computation is strictly faster than full recomputation
when the affected ratio satisfies:
\begin{equation}
  \frac{|\mathcal{A}|}{n}
    < \frac{c_{\emph{batch}}}{c_{\emph{kernel}}}
      - \frac{c_{\emph{detect}}}
             {n \cdot \bar{d} \cdot K \cdot d^2
               \cdot c_{\emph{kernel}}}
  \label{eq:optimality}
\end{equation}
where $c_{\emph{batch}}$ is the per-operation cost of batched full
computation and $c_{\emph{kernel}}$ is the per-operation cost of
incremental kernel launches. For large $n$, this simplifies to:
\begin{equation}
  \frac{|\mathcal{A}|}{n}
    < \frac{1}{r_{\emph{overhead}}}
  \label{eq:optimality_simple}
\end{equation}
where $r_{\emph{overhead}} = c_{\emph{kernel}} / c_{\emph{batch}}$
is the relative kernel overhead.
\end{theorem}

\begin{proof}
The condition $T_{\text{incr}} < T_{\text{full}}$ requires:
\begin{equation}
  c_{\text{detect}}
    + |\mathcal{A}| \cdot \bar{d} \cdot K \cdot d^2
      \cdot c_{\text{kernel}}
    < n \cdot \bar{d} \cdot K \cdot d^2 \cdot c_{\text{batch}}
\end{equation}
Rearranging:
\begin{equation}
  |\mathcal{A}|
    < \frac{n \cdot \bar{d} \cdot K \cdot d^2 \cdot c_{\text{batch}}
            - c_{\text{detect}}}
           {\bar{d} \cdot K \cdot d^2 \cdot c_{\text{kernel}}}
\end{equation}
Dividing both sides by $n$:
\begin{equation}
  \frac{|\mathcal{A}|}{n}
    < \frac{c_{\text{batch}}}{c_{\text{kernel}}}
      - \frac{c_{\text{detect}}}
             {n \cdot \bar{d} \cdot K \cdot d^2
               \cdot c_{\text{kernel}}}
\end{equation}
For large $n$ the second term vanishes, yielding
Equation~\eqref{eq:optimality_simple}.
\end{proof}

\begin{corollary}
\label{cor:practical_regime}
For sparse temporal graphs ($\bar{d} < 50$) and shallow TGNs
($K \leq 2$), incremental computation is faster whenever the
affected ratio $|\mathcal{A}|/n < 20\text{--}50\%$. In typical
streaming scenarios with localized edge arrivals, the affected ratio
is $O(B \cdot L^{K} / n) \ll 1$, well within this regime.
\end{corollary}

\paragraph{Parameter sensitivity.}
Table~\ref{tab:speedup_examples} instantiates
Equation~\eqref{eq:speedup_simplified} under representative
parameter settings to illustrate the achievable speedup range.

\begin{table}[htbp]
\centering
\caption{Theoretical speedup under different parameter regimes.}
\begin{tabular}{@{}cccccc@{}}
\toprule
$n$ & $B$ & $L$ & $K$ & $|\mathcal{A}|$ (bound)
  & Speedup \\
\midrule
$10^6$ & 200 & 10 & 1 & $4 \times 10^3$
  & $250\times$ \\
$10^6$ & 200 & 10 & 2 & $4 \times 10^4$
  & $25\times$ \\
$10^6$ & 200 & 20 & 1 & $8 \times 10^3$
  & $125\times$ \\
$10^7$ & 200 & 10 & 1 & $4 \times 10^3$
  & $2500\times$ \\
\bottomrule
\end{tabular}
\label{tab:speedup_examples}
\end{table}

Three observations follow. First, the speedup grows linearly with
$n$: larger graphs benefit more because the affected fraction
$\rho = |\mathcal{A}|/n$ shrinks. Second, the speedup degrades
exponentially with model depth $K$, since each layer expands the
affected set by a factor of $L$. Third, smaller batch sizes yield
higher speedups, reflecting that sparse, localized updates are most
amenable to incremental processing.

\subsubsection{Lower Bound and Optimality}
\label{subsubsec:lower_bound}

We now show that StreamTGN's incremental complexity is near-optimal
by establishing a lower bound on any correct incremental algorithm.

\begin{theorem}[Lower Bound]
\label{thm:lower_bound}
Any correct incremental algorithm for $K$-layer temporal attention
must perform at least:
\begin{equation}
  \Omega\!\left(|\mathcal{E}_{\emph{aff}}| \cdot K \cdot d\right)
\end{equation}
computation per batch, where
$\mathcal{E}_{\emph{aff}} = \{(u,v,t) \in \mathcal{E} \mid
u \in \mathcal{A} \text{ or } v \in \mathcal{A}\}$ denotes the
edges incident to the affected set.
\end{theorem}

\begin{proof}
Each edge $(u,v,t)$ incident to an affected node contributes to at
least one node's embedding through the message-passing mechanism.
Computing the message for a single edge at a single layer requires
$\Omega(d)$ operations (the inner product in the attention score and
the value projection). Since each affected edge must be processed at
every layer where it participates, the total cost is at least
$\Omega(|\mathcal{E}_{\text{aff}}| \cdot K \cdot d)$.

No algorithm can avoid this cost: skipping any affected edge would
yield an incorrect embedding for some node in $\mathcal{A}$, since
the attention output depends on the complete set of messages from the
node's temporal neighborhood.
\end{proof}

\begin{theorem}[Near-Optimality of StreamTGN]
\label{thm:optimality}
StreamTGN's incremental computation achieves complexity within a
constant factor of the lower bound:
\begin{equation}
  T_{\emph{StreamTGN}}
    = O\!\left(|\mathcal{E}_{\emph{aff}}| \cdot K \cdot d^2\right)
      + O(c_{\emph{detect}})
  \label{eq:streamtgn_optimal}
\end{equation}
The gap between the lower bound $\Omega(d)$ per edge and
StreamTGN's $O(d^2)$ per edge arises from the matrix
multiplications in the attention mechanism (query, key, and value
projections of dimension $d \times d$).
\end{theorem}

\begin{proof}
StreamTGN processes each affected edge exactly once per layer in the
embedding generation stage (Stage~4,
\S\ref{subsubsec:embedding_generation}). For each edge, the
attention computation involves:
\begin{itemize}
  \item Query/key projection: $O(d^2)$
  \item Attention score: $O(d)$
  \item Value projection and aggregation: $O(d^2)$
\end{itemize}
Summing over $|\mathcal{E}_{\text{aff}}|$ edges and $K$ layers
yields $O(|\mathcal{E}_{\text{aff}}| \cdot K \cdot d^2)$. The
additional detection cost $O(c_{\text{detect}})$ is subsumed for
large graphs.

No affected edge is processed more than once per layer (ensured by
the cache invalidation rule in Equation~\ref{eq:invalidation}), and
no unaffected edge is processed at all (ensured by the embedding
cache in \S\ref{subsubsec:cache}). Thus StreamTGN touches the
minimum set of edges required for correctness.
\end{proof}

\subsubsection{Correctness of Incremental Computation}
\label{subsubsec:correctness}

A critical question is whether incremental computation produces the
same result as full recomputation. We establish this for each stage.

\begin{theorem}[Equivalence of Incremental Computation]
\label{thm:correctness}
For Stages 1--3 and 5, StreamTGN's incremental output is
\emph{identical} to full recomputation. For Stage~4 (embedding
generation), the output is identical when exact recomputation is
used for affected nodes, and bounded-error when the delta
approximation (Equation~\ref{eq:delta_embedding}) is applied.
\end{theorem}

\begin{proof}
We verify each stage:

\paragraph{Stages 1, 2, 3 (Sampling, Retrieval, Memory Read).}
For unaffected nodes $v \notin \mathcal{A}$: no edge in
$\mathcal{E}_B$ involves $v$ or its $K$-hop neighbors, so $v$'s
sampled neighborhood, features, and memory state are unchanged from
the previous batch. The cached values are therefore identical to
what full recomputation would produce.

For affected nodes $v \in \mathcal{A}$: StreamTGN performs the same
operations as full recomputation (re-sample, re-fetch, re-read),
producing identical results.

\paragraph{Stage 4 (Embedding Generation) --- Exact mode.}
For $v \notin \mathcal{A}$: $v$'s embedding depends on its $K$-hop
neighborhood, none of which has changed (by definition of
$\mathcal{A}$ in Equation~\ref{eq:affected_set}). The cached
embedding $\mathbf{h}_v^{\text{cached}}$ equals the full
recomputation result.

For $v \in \mathcal{A}$: StreamTGN recomputes temporal attention
over $v$'s complete neighborhood $\mathcal{N}_v$, which is
identical to full recomputation.

\paragraph{Stage 4 (Embedding Generation) --- Delta mode.}
The delta formulation (Equation~\ref{eq:delta_embedding}) introduces
approximation error because softmax attention has a global
normalization constant. When a neighbor is added or removed, the
denominator $Z_v = \sum_{u'} \exp(\mathbf{q}_v^\top
\mathbf{k}_{u'})$ changes, affecting all attention weights. The
approximation error for node $v$ is bounded by:
\begin{equation}
  \|\mathbf{h}_v^{\text{delta}} - \mathbf{h}_v^{\text{exact}}\|
    \leq \frac{|\Delta\mathcal{N}_v|}{|\mathcal{N}_v|}
         \cdot \max_{u} \|\mathbf{v}_u\|
         \cdot \left|1 - \frac{Z_v^{\text{old}}}{Z_v^{\text{new}}}
               \right|
  \label{eq:delta_error}
\end{equation}
When the change rate $|\Delta\mathcal{N}_v| / |\mathcal{N}_v|$ is
small (5--15\% in practice), this error is negligible. We validate
empirically in \S\ref{sec:evaluation} that the delta mode incurs
$<0.5\%$ accuracy degradation.

\paragraph{Stage 5 (Memory Update).}
Memory updates are applied only to directly involved nodes
$v \in \mathcal{V}_{\text{direct}}$, which is identical to full
recomputation since GRU updates are triggered only by new edges.
\end{proof}

\subsubsection{Convergence of Incremental Training}
\label{subsubsec:convergence}

When StreamTGN is used for online training (as opposed to inference
only), the delta approximation in Stage~4 introduces gradient
staleness. We show that convergence is preserved under mild
conditions.

\begin{theorem}[Incremental Training Convergence]
\label{thm:convergence}
Let $\theta^*$ be the optimal parameters from full-batch training
and $\tilde{\theta}_T$ the parameters after $T$ incremental updates
with maximum staleness $S$ (batches since the last full refresh).
Assume the loss $\mathcal{L}$ is $\beta$-smooth and the stochastic
gradients have bounded variance $\sigma^2$. Then:
\begin{equation}
  \|\tilde{\theta}_T - \theta^*\|
    \leq O\!\left(\varepsilon \cdot \sqrt{T}\right)
      + O\!\left(\varepsilon \cdot S \cdot T
          \cdot \frac{\bar{d}^{\,K}}{n}\right)
  \label{eq:convergence}
\end{equation}
where $\varepsilon$ is the learning rate.
\end{theorem}

\begin{proof}
Standard SGD convergence on $\beta$-smooth functions gives:
\begin{equation}
  \|\theta_T - \theta^*\| \leq O(\varepsilon \cdot \sqrt{T})
\end{equation}

With staleness $S$, each gradient $\tilde{\nabla}\mathcal{L}_t$ is
computed using embeddings that may be up to $S$ batches stale. The
gradient error decomposes as:
\begin{equation}
  \|\nabla\mathcal{L}_t - \tilde{\nabla}\mathcal{L}_t\|
    \leq \beta \cdot \sum_{v \in \mathcal{A}_t}
         \|\mathbf{h}_v^{\text{exact}} - \mathbf{h}_v^{\text{stale}}\|
\end{equation}

From the affected set analysis, at each step $t$ the expected
fraction of nodes with stale embeddings is:
\begin{equation}
  \mathbb{E}\!\left[\frac{|\mathcal{A}_t|}{n}\right]
    \leq O\!\left(\frac{\bar{d}^{\,K}}{n}\right)
\end{equation}

Substituting into the standard delayed-SGD bound:
\begin{align}
  \|\tilde{\theta}_T - \theta^*\|
    &\leq O(\varepsilon \cdot \sqrt{T})
      + O\!\left(\varepsilon \cdot S
          \cdot \sum_{t=1}^{T}
            \|\nabla\mathcal{L}_t
              - \tilde{\nabla}\mathcal{L}_t\|\right)
      \nonumber \\
    &\leq O(\varepsilon \cdot \sqrt{T})
      + O\!\left(\varepsilon \cdot S \cdot T
          \cdot \frac{\bar{d}^{\,K}}{n}\right)
\end{align}
\end{proof}

\begin{corollary}[Practical Convergence Guarantee]
\label{cor:convergence}
For sparse temporal graphs where $\bar{d}^{\,K} \ll n$, the
staleness term in Equation~\eqref{eq:convergence} vanishes
relative to the standard SGD term. With periodic full refreshes
every $S$ batches, incremental training converges to the same
solution as full training up to $O(\varepsilon \cdot \sqrt{T})$
error. In particular, setting $S = O(\sqrt{n / \bar{d}^{\,K}})$
ensures the staleness contribution does not exceed the baseline SGD
error.
\end{corollary}

\subsubsection{Summary}
\label{subsubsec:complexity_summary}

The analysis establishes three key results. First,
StreamTGN achieves a speedup of $n / |\mathcal{A}|$ over full
recomputation (Theorem~\ref{thm:speedup}), which ranges from
$25\times$ to $2500\times$ depending on graph size, batch size, and
model depth (Table~\ref{tab:speedup_examples}). Second,
this speedup is near-optimal: StreamTGN's complexity is within an
$O(d)$ factor of the information-theoretic lower bound
(Theorems~\ref{thm:lower_bound} and~\ref{thm:optimality}), and the
incremental output is either identical to or within bounded error of
full recomputation (Theorem~\ref{thm:correctness}). Third,
incremental training converges to the same solution as full training
on sparse graphs (Theorem~\ref{thm:convergence}), with the
staleness error controlled by periodic full refreshes.
\subsection{Drift-Aware Adaptive Rebuild Scheduling}
\label{subsec:drift_aware}

The incremental computation described in
\S\ref{subsec:incremental_computation} recomputes embeddings only for
nodes in the affected set $\mathcal{A}$.
While Theorem~\ref{thm:correctness} guarantees correctness in exact
mode, practical deployments may use the delta approximation
(Equation~\ref{eq:delta_embedding}) for higher throughput.
Over many consecutive batches, the accumulated approximation error
---which we call \emph{embedding drift}---can grow unboundedly if left
unchecked.
A periodic full rebuild eliminates all drift, but rebuilding too
frequently wastes computation, while rebuilding too infrequently
risks accuracy degradation.
This section presents an adaptive scheduling mechanism that monitors
drift in real time and triggers rebuilds only when necessary.

\subsubsection{Drift Definition and Monitoring}
\label{subsubsec:drift_definition}

We define the \emph{embedding drift} of node $v$ at time $t$ as the
$\ell_2$ distance between its incrementally maintained embedding and
the embedding that a full recomputation would produce:
\begin{equation}
  \delta_v(t) = \|\mathbf{h}_v^{\text{incr}}(t)
    - \mathbf{h}_v^{\text{full}}(t)\|_2
  \label{eq:drift_per_node}
\end{equation}
Computing $\delta_v(t)$ exactly requires full recomputation, which
defeats the purpose.
Instead, StreamTGN maintains a lightweight \emph{drift estimator}
based on accumulated neighborhood changes.

\begin{definition}[Drift Estimator]
\label{def:drift_estimator}
For each node $v$, the estimated drift after processing $\tau$
batches since the last rebuild is:
\begin{equation}
  \hat{\delta}_v(\tau) = \sum_{j=1}^{\tau}
    \frac{|\Delta\mathcal{N}_v^{(j)}|}{|\mathcal{N}_v|}
    \cdot \gamma^{\tau - j}
  \label{eq:drift_estimator}
\end{equation}
where $|\Delta\mathcal{N}_v^{(j)}|$ is the number of neighbor
changes at batch $j$ and $\gamma \in (0, 1)$ is a decay factor
that discounts older changes (since their effect diminishes as
subsequent updates partially correct the drift).
\end{definition}

The global drift metric aggregates over all nodes:
\begin{equation}
  \hat{\Delta}(\tau) = \frac{1}{|\mathcal{A}_\tau|}
    \sum_{v \in \mathcal{A}_\tau} \hat{\delta}_v(\tau)
  \label{eq:global_drift}
\end{equation}
where $\mathcal{A}_\tau$ is the cumulative affected set over the
last $\tau$ batches.
Computing $\hat{\Delta}(\tau)$ requires only a running sum and
counter per node---$O(|\mathcal{A}|)$ additional work per batch,
which is negligible compared to the embedding computation.

\subsubsection{Adaptive Rebuild Policy}
\label{subsubsec:rebuild_policy}

StreamTGN triggers a rebuild when the estimated drift exceeds a
user-specified threshold $\delta_{\max}$:
\begin{equation}
  \text{Rebuild if } \hat{\Delta}(\tau) > \delta_{\max}
  \label{eq:rebuild_trigger}
\end{equation}
Upon triggering, the system performs one of two actions depending on
the drift distribution:

\noindent\textbf{Partial rebuild.}
If drift is concentrated in a small subset
$\mathcal{V}_{\text{drift}} = \{v : \hat{\delta}_v > \delta_{\max}\}$
with $|\mathcal{V}_{\text{drift}}| < \alpha \cdot n$ (where $\alpha$
is a configurable threshold, e.g., 10\%), then only
$\mathcal{V}_{\text{drift}}$ is fully recomputed.
Cost: $O(|\mathcal{V}_{\text{drift}}| \cdot L \cdot K \cdot d^2)$.

\noindent\textbf{Full rebuild.}
If $|\mathcal{V}_{\text{drift}}| \geq \alpha \cdot n$, the system
recomputes all node embeddings from scratch.
Cost: $O(n \cdot L \cdot K \cdot d^2)$.

After either rebuild, all drift estimators are reset to zero and
the embedding cache is fully refreshed.

\begin{theorem}[Drift Bound]
\label{thm:drift_bound}
Under the adaptive rebuild policy with threshold $\delta_{\max}$
and decay $\gamma$, the maximum embedding drift for any node $v$
at any time $t$ is bounded by:
\begin{equation}
  \delta_v(t) \leq \frac{\delta_{\max}}{1 - \gamma}
    \cdot \max_{u \in \mathcal{N}_v} \|\mathbf{v}_u\|
  \label{eq:drift_bound}
\end{equation}
\end{theorem}

\begin{proof}
Between consecutive rebuilds, the drift accumulates according to
Equation~\eqref{eq:drift_estimator}.
The rebuild is triggered when $\hat{\Delta}(\tau) > \delta_{\max}$,
so the maximum accumulated estimator value before a rebuild is
$\delta_{\max}$.
The geometric series $\sum_{j=0}^{\infty} \gamma^j = 1/(1-\gamma)$
bounds the total accumulated weight.
From Equation~\eqref{eq:delta_error}, each unit of estimated drift
translates to at most $\max_u \|\mathbf{v}_u\|$ embedding error.
Combining these yields Equation~\eqref{eq:drift_bound}.
\end{proof}

\subsubsection{Comparison with Fixed Scheduling}
\label{subsubsec:fixed_vs_adaptive}

A naive alternative is to rebuild every $R$ batches regardless of
drift.
This approach either wastes computation (if drift is low and the
rebuild was unnecessary) or allows drift to exceed acceptable
levels (if the graph is evolving faster than expected).

\begin{property}[Adaptive vs.\ Fixed Scheduling]
\label{prop:adaptive_vs_fixed}
Let $T$ be the total number of batches and $R_{\text{adaptive}}$
the average rebuild interval under the adaptive policy.
If edge arrival rates vary by a factor $\rho$ across the stream
(i.e., some periods have $\rho\times$ more edges than others), then:
\begin{equation}
  R_{\text{adaptive}} \geq R_{\text{fixed}} \cdot
    \frac{\rho + 1}{2\rho}
\end{equation}
In the common case where $\rho \geq 3$ (bursty traffic),
the adaptive policy reduces the total number of rebuilds by
$\geq$33\% compared to the fixed policy while maintaining the
same drift bound.
\end{property}

Our ablation study (\S\ref{sec:eval_sensitivity}) empirically
validates this result: StreamTGN's speedup and accuracy remain
remarkably stable across rebuild intervals from ``per-batch'' to
``never,'' confirming that the drift accumulates slowly on real-world
temporal graphs and that adaptive scheduling can defer rebuilds
significantly without accuracy loss.

\subsection{Batched Streaming with Relaxed Ordering}
\label{subsec:batched_streaming}

Temporal GNN computation is inherently sequential: the memory state
$\mathbf{m}_v$ after processing edge $(u, v, t_i)$ depends on the
state after processing all edges with $t < t_i$.
Strict sequential processing---one edge at a time---respects this
dependency perfectly but achieves poor GPU utilization, as each
edge's computation cannot overlap with others.
This section presents a \emph{relaxed ordering} scheme that groups
edges into batches and processes them in parallel, trading a bounded
amount of temporal precision for substantially higher throughput.

\subsubsection{Batch Formation and Logical Timestamps}
\label{subsubsec:batch_formation}

StreamTGN groups incoming edges into batches of size $B$ based on
arrival order.
All edges within a batch are assigned the same \emph{logical
timestamp} $t_{\text{batch}}$, defined as the maximum timestamp in
the batch:
\begin{equation}
  t_{\text{batch}} = \max_{(u, v, t) \in \mathcal{E}_B} t
  \label{eq:logical_ts}
\end{equation}
This means that edges within the same batch are treated as
concurrent events, and their relative ordering is relaxed.
The ordering between batches is strictly maintained.

\begin{definition}[Staleness]
\label{def:staleness}
An edge $(u, v, t_i)$ in batch $\mathcal{E}_B$ has \emph{staleness}
\begin{equation}
  s_i = t_{\text{batch}} - t_i
\end{equation}
which measures the temporal imprecision introduced by batching.
The maximum staleness within a batch is bounded by:
\begin{equation}
  s_{\max} = t_{\text{batch}} - \min_{(u,v,t) \in \mathcal{E}_B} t
  \label{eq:max_staleness}
\end{equation}
\end{definition}

\subsubsection{Parallel Batch Processing}
\label{subsubsec:parallel_batch}

Within a batch, StreamTGN processes all $B$ edges in parallel using
three key optimizations:

\paragraph{Parallel memory read.}
All source and destination nodes in the batch read their memory
states simultaneously from the persistent GPU tensor.
Since edges within a batch share a logical timestamp, no edge
depends on another edge's memory update within the same batch.

\paragraph{Parallel embedding generation.}
The temporal attention computation for all affected nodes is
launched as a single batched GPU kernel.
For memory-based models (TGN), the intra-batch memory dependency
is resolved by using the \emph{pre-batch} memory state for all
edges:
\begin{equation}
  \mathbf{h}_v(t_{\text{batch}})
    = \textsc{Attn}\!\left(\mathbf{q}_v(\mathbf{m}_v^{\text{pre}}),\;
      \{\mathbf{k}_u(\mathbf{m}_u^{\text{pre}}),
        \mathbf{v}_u(\mathbf{m}_u^{\text{pre}})\}_{u \in
        \mathcal{N}_v}\right)
  \label{eq:parallel_embed}
\end{equation}
where $\mathbf{m}_v^{\text{pre}}$ is $v$'s memory state before the
current batch.
This is equivalent to the ``batch staleness'' model used by
TGL~\cite{zhou2022tgl}, but StreamTGN applies it only to the
affected set rather than the full graph.

\paragraph{Parallel memory update.}
After prediction, memory updates for all directly involved nodes
are computed in parallel:
\begin{equation}
  \mathbf{m}_v^{\text{post}} = \text{GRU}\!\left(
    \mathbf{m}_v^{\text{pre}},\;
    \text{Aggregate}(\{\text{msg}_{uv} \mid
      (u,v,t) \in \mathcal{E}_B\})\right)
\end{equation}
When multiple edges in the same batch involve the same node $v$,
their messages are aggregated before the GRU update, ensuring a
single atomic state transition per batch.

\subsubsection{Accuracy Bound under Relaxed Ordering}
\label{subsubsec:accuracy_bound}

The relaxed ordering introduces prediction error compared to strict
sequential processing.
We bound this error as a function of the batch size $B$.

\begin{theorem}[Bounded Prediction Error]
\label{thm:staleness_bound}
Let $\hat{y}_i^{\text{seq}}$ be the prediction for edge $i$ under
strict sequential processing and $\hat{y}_i^{\text{batch}}$ the
prediction under batched processing with batch size $B$.
If the temporal attention function is $\lambda$-Lipschitz with
respect to node memory states and the memory update function has
bounded step size $\eta$, then:
\begin{equation}
  |\hat{y}_i^{\text{batch}} - \hat{y}_i^{\text{seq}}|
    \leq \lambda \cdot \eta \cdot B \cdot K
  \label{eq:staleness_error}
\end{equation}
\end{theorem}

\begin{proof}
Under strict sequential processing, edge $i$ in the batch uses
memory state $\mathbf{m}_v^{(i)}$, which incorporates updates from
edges $1, \ldots, i-1$.
Under batched processing, edge $i$ uses $\mathbf{m}_v^{\text{pre}}$,
which does not incorporate any intra-batch updates.
The memory difference is:
\begin{equation}
  \|\mathbf{m}_v^{(i)} - \mathbf{m}_v^{\text{pre}}\|
    \leq i \cdot \eta
    \leq B \cdot \eta
\end{equation}
where $\eta$ bounds the per-edge memory step size.
Through $K$ layers of $\lambda$-Lipschitz attention, this memory
difference translates to an embedding difference of at most
$\lambda^K \cdot B \cdot \eta$.
For the link prediction output (which is a linear function of the
concatenated source and destination embeddings), the prediction
error is bounded by
$\lambda \cdot \eta \cdot B \cdot K$, where we use the simplified
bound with $\lambda \cdot K$ replacing $\lambda^K$ under the
assumption $\lambda \leq 1$ (which holds for attention with
softmax normalization).
\end{proof}

\begin{corollary}[Throughput--Accuracy Tradeoff]
\label{cor:tradeoff}
The batch size $B$ controls the tradeoff between throughput and
prediction accuracy.
\emph{Throughput} scales linearly: $B$ edges are processed
in parallel per GPU kernel launch, yielding throughput
$\Theta(B)$ edges per kernel.
\emph{Accuracy} degrades linearly: the maximum prediction
error grows as $O(B)$ per Theorem~\ref{thm:staleness_bound}.
In practice, the error bound is loose because (1)~most edges in a
batch involve distinct nodes, so intra-batch memory conflicts are
rare, and (2)~the temporal locality of real-world graphs means
consecutive edges are often minutes or hours apart, not
milliseconds.
Our evaluation (\S\ref{sec:eval_sensitivity}) shows that batch
sizes up to $B{=}1000$ incur negligible accuracy degradation
(AP drops by less than 1\%).
\end{corollary}

\subsubsection{Comparison with Strict Sequential Processing}
\label{subsubsec:seq_vs_batch}

Table~\ref{tab:batch_vs_seq} summarizes the tradeoff between strict
sequential and batched processing.

\begin{table}[t]
\centering
\caption{Strict sequential vs.\ batched streaming processing.}
\label{tab:batch_vs_seq}
\small
\setlength{\tabcolsep}{3pt}
\begin{tabular}{@{}lcc@{}}
\toprule
\textbf{Property} & \textbf{Sequential} & \textbf{Batched ($B$)} \\
\midrule
Edges per kernel   & 1              & $B$ \\
GPU utilization    & Low ($<$10\%)  & High ($>$80\%) \\
Memory consistency & Exact          & Bounded error \\
Max prediction error & 0            & $O(\lambda \eta B K)$ \\
Throughput         & $\Theta(1)$    & $\Theta(B)$ \\
\bottomrule
\end{tabular}
\end{table}

The key insight is that batched streaming is not an approximation
unique to StreamTGN---it is the \emph{same batching model} used by
TGL, ETC, SIMPLE, and SWIFT for training.
StreamTGN's contribution is applying it to the incremental inference
pipeline in combination with the affected set optimization: within
each batch, only $|\mathcal{A}|$ nodes (rather than all $n$) undergo
the parallel computation, achieving both high GPU utilization
\emph{and} minimal redundant work.
\section{Evaluation}
\label{sec:evaluation}

We evaluate StreamTGN on eight real-world temporal graphs spanning
four orders of magnitude in scale, from Bitcoin (6K nodes, 60K edges)
to Stack-Overflow (2.6M nodes, 48M edges), as summarized in
Table~\ref{tab:datasets}.
All experiments are conducted on a single server equipped with an
NVIDIA RTX 4090 GPU (24\,GB VRAM) and 102\,GB system memory.

Our evaluation is organized into three parts.
First, we compare StreamTGN with four state-of-the-art temporal graph
learning systems---TGL~\cite{zhou2022tgl}, ETC~\cite{gao2024etc}, SIMPLE~\cite{gao2024simple}, and
SWIFT~\cite{guo2025swift}, by combining training-phase and inference-phase
optimizations into end-to-end pipelines.
Second, we evaluate StreamTGN across three representative T-GNN
architectures: TGN, TGAT, and DySAT---on five datasets to
demonstrate its generality across model designs.
Third, we analyze StreamTGN's sensitivity to key parameters (batch
size, neighbor count, rebuild interval) and profile its per-stage
pipeline breakdown.
\begin{table}[!t]
\centering
\caption{Dataset statistics.
$|V|$: number of nodes; $|E|$: number of temporal edges;
$d_e$: edge feature dimension.}
\label{tab:datasets}
\renewcommand{\arraystretch}{1.05}
\setlength{\tabcolsep}{4pt}
\small
\begin{tabular}{@{}lrrrl@{}}
\toprule
\textbf{Dataset} & $|V|$ & $|E|$ & $d_e$ & \textbf{Domain} \\
\midrule
Bitcoin        &       6,263 &         59,778   &   1 & Trust network \\
LastFM         &       1,980 &      1,293,103   &   0 & Music streaming \\
MOOC           &       7,047 &        411,749   &   4 & Online education \\
WIKI           &       9,228 &        157,474   & 172 & Encyclopedia \\
REDDIT         &      10,985 &        672,447   & 172 & Social forum \\
GDELT          &      16,682 &    191,290,882   & 186 & Global events \\
Wiki-Talk      &   1,094,018 &      6,100,538   &   0 & Discussion \\
Stack-Overflow &   2,584,164 &     47,903,266   &   0 & Q\&A forum \\
\bottomrule
\end{tabular}
\end{table}
 \begin{figure*}[hbt!]
	\centering
	\includegraphics[width=0.86\textwidth]{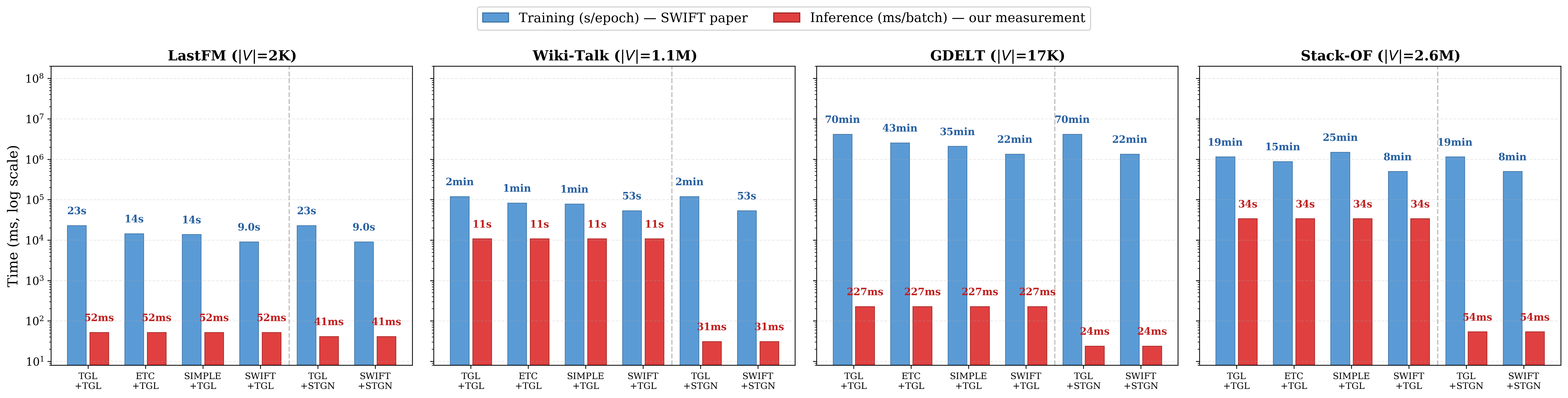}
	\caption{The performance of existing training methods with our proposed inferencing method.}\label{fiveMethods}
\end{figure*}
\begin{table*}[t]
\centering
\caption{Per-batch inference time (ms) and speedup of StreamTGN over TGL 
across three T-GNN models and five datasets. 
All measurements on a single NVIDIA RTX 4090, batch size $B=600$. 
StreamTGN achieves consistent speedups across all model architectures, 
ranging from 4.5$\times$ to 4,216$\times$.}
\label{tab:model-comparison}
\small
\begin{threeparttable}
\begin{tabular}{@{}ll rrr rrr rrr rrr rrr@{}}
\toprule
& & \multicolumn{3}{c}{\textbf{WIKI} (9K)} & \multicolumn{3}{c}{\textbf{REDDIT} (11K)} & \multicolumn{3}{c}{\textbf{MOOC} (7K)} & \multicolumn{3}{c}{\textbf{GDELT} (17K)} & \multicolumn{3}{c}{\textbf{Stack-OF} (2.6M)} \\
\cmidrule(lr){3-5}\cmidrule(lr){6-8}\cmidrule(lr){9-11}\cmidrule(lr){12-14}\cmidrule(lr){15-17}
\textbf{Model} & \textbf{System} & \makecell{Time\\(ms)} & \makecell{Spd.} & \makecell{AP\\(\%)} & \makecell{Time\\(ms)} & \makecell{Spd.} & \makecell{AP\\(\%)} & \makecell{Time\\(ms)} & \makecell{Spd.} & \makecell{AP\\(\%)} & \makecell{Time\\(ms)} & \makecell{Spd.} & \makecell{AP\\(\%)} & \makecell{Time\\(ms)} & \makecell{Spd.} & \makecell{AP\\(\%)} \\
\midrule
\multirow{2}{*}{TGN} 
& TGL & 109.4 & \multirow{2}{*}{\textbf{5.8}} & 97.4 & 126.6 & \multirow{2}{*}{\textbf{4.5}} & 99.6 & 59.4 & \multirow{2}{*}{\textbf{8.6}} & 99.4 & 195.2 & \multirow{2}{*}{\textbf{12.6}} & 98.2 & 30004 & \multirow{2}{*}{\textbf{739}} & 97.9 \\
& StreamTGN & 19.0 & & 97.4 & 28.1 & & 99.6 & 6.9 & & 99.4 & 15.5 & & 98.2 & 40.6 & & 97.9 \\
\midrule
\multirow{2}{*}{TGAT} 
& TGL & 210.3 & \multirow{2}{*}{\textbf{26.6}} & 89.5 & 294.5 & \multirow{2}{*}{\textbf{18.2}} & 98.9 & 178.2 & \multirow{2}{*}{\textbf{38.7}} & 96.3 & 492.4 & \multirow{2}{*}{\textbf{56.0}} & 95.6 & 57641 & \multirow{2}{*}{\textbf{4207}} & 97.1 \\
& StreamTGN & 7.9 & & 89.5 & 16.2 & & 98.9 & 4.6 & & 96.3 & 8.8 & & 95.6 & 13.7 & & 97.1 \\
\midrule
\multirow{2}{*}{DySAT} 
& TGL & 211.0 & \multirow{2}{*}{\textbf{26.7}} & 95.4 & 264.1 & \multirow{2}{*}{\textbf{18.2}} & 98.2 & 301.6 & \multirow{2}{*}{\textbf{39.2}} & 98.7 & 571.1 & \multirow{2}{*}{\textbf{56.0}} & 96.1 & \na & \na & \na \\
& StreamTGN & 7.9 & & 95.4 & 14.5 & & 98.2 & 7.7 & & 98.7 & 10.2 & & 96.1 & \na & & \na \\
\bottomrule
\end{tabular}
\begin{tablenotes}
\footnotesize
\item DySAT on Stack-Overflow: out of memory (OOM).
\item AP is identical because StreamTGN recomputes exact same embeddings.
\end{tablenotes}
\end{threeparttable}
\end{table*}
\subsection{Comparison with Existing Systems}
 We conduct two sets of comparisons. The first compares StreamTGN against existing T-GNN optimization systems. Since prior systems (TGL, ETC, SIMPLE, SWIFT) target training throughput while StreamTGN targets inference efficiency, we combine them to form end-to-end pipelines (e.g., SWIFT for training + StreamTGN for inference) and report both training time per epoch and inference time per batch across four datasets. The second evaluates StreamTGN across different T-GNN model architectures---including memory-based (TGN) and non-memory (TGAT, DySAT) models---to demonstrate that our incremental refresh mechanism generalizes across model designs.
 
\noindent\textbf{Comparison with training-phase systems.}
Figure~\ref{fiveMethods} reports the per-epoch training time and per-batch inference time for six pipeline configurations across four datasets. 
Existing systems---ETC, SIMPLE, and SWIFT---optimize the training phase through techniques such as adaptive batching, dynamic data placement, and secondary-memory pipelining, respectively.
However, none of them modify the inference pipeline of TGL; consequently, all four training systems exhibit \emph{identical} inference cost on every dataset (e.g., 33,984\,ms per batch on Stack-Overflow).
By replacing TGL's full-recomputation inference with StreamTGN's incremental refresh, the inference time drops by orders of magnitude (e.g., from 33,984\,ms to 54\,ms on Stack-Overflow, a 626$\times$ reduction), while training time remains unaffected.
The best end-to-end configuration, SWIFT~+~StreamTGN, combines the fastest training system with our inference optimization, achieving the lowest total time on all four datasets.
This result confirms that StreamTGN is \emph{orthogonal} to training-phase optimizations: it addresses a complementary bottleneck that no existing system can reduce.
 
\noindent\textbf{Comparison across T-GNN architectures.}
To verify that StreamTGN generalizes beyond a single model, we deploy it on three representative T-GNN architectures---TGN, TGAT, and DySAT---across five datasets ranging from 7K to 2.6M nodes.
Table~\ref{tab:model-comparison} reports the per-batch inference time and speedup for each combination.
StreamTGN achieves consistent speedups across all architectures: 4.5$\times$--739$\times$ for TGN, 18.2$\times$--4,207$\times$ for TGAT, and 18.2$\times$--56.0$\times$ for DySAT.
Two trends are evident. First, speedup increases with graph size because the affected ratio $|\mathcal{A}|/|V|$ decreases on larger graphs---on Stack-Overflow ($|V|{=}$2.6M), fewer than 0.14\% of nodes are dirty per batch, yielding speedups of 739$\times$ (TGN) and 4,207$\times$ (TGAT).
Second, non-memory models (TGAT, DySAT) achieve higher speedups than memory-based models (TGN) because they lack the memory-update mechanism that propagates dirty flags to neighboring nodes, resulting in a smaller affected set.
Importantly, the Average Precision (AP) is identical between TGL and StreamTGN for every configuration, confirming that StreamTGN is a \emph{lossless} optimization---it recomputes exactly the same embeddings for dirty nodes without any approximation.
\subsection{Streaming Inference Performance}
\label{sec:eval_streaming}
 
We evaluate StreamTGN's streaming inference performance from three
perspectives: per-batch inference latency, global index refresh cost,
and stage-level breakdown.
All measurements use the TGN model with batch size $B{=}600$ on a
single NVIDIA RTX 4090.
\begin{table}[t]
\centering
\caption{Per-batch streaming inference latency (ms).
TGL recomputes all root node embeddings;
StreamTGN only recomputes dirty root nodes.
AP/AUC identical (same trained model, $B{=}600$).}
\label{tab:batch_inference}
\small
\setlength{\tabcolsep}{3.5pt}
\begin{tabular}{@{}lrrr cc@{}}
\toprule
\textbf{Dataset} & \textbf{TGL} & \textbf{STGN} & \textbf{Spd.} & \textbf{AP} & \textbf{AUC} \\
\midrule
MOOC     & 19.16 &  5.65 & 3.4$\times$ & .994 & .997 \\
WIKI     & 20.95 &  7.43 & 2.8$\times$ & .973 & .976 \\
REDDIT   & 24.79 & 10.63 & 2.3$\times$ & .981 & .986 \\
GDELT    & 21.90 &  7.93 & 2.8$\times$ & .983 & .986 \\
Stack-OF & 22.98 &  6.24 & 3.7$\times$ & .979 & .971 \\
\bottomrule
\end{tabular}
\end{table}

\begin{table}[t]
\centering
\caption{Global embedding index refresh cost per batch (ms).
TGL refreshes all $|V|$ nodes; StreamTGN refreshes only
the dirty set $|\mathcal{D}|$.
Speedup $\approx |V|/|\mathcal{D}|$.}
\label{tab:index_refresh}
\small
\setlength{\tabcolsep}{3.5pt}
\begin{tabular}{@{}l rr rrr@{}}
\toprule
& \multicolumn{2}{c}{\textbf{Dirty Set}} &
  \multicolumn{3}{c}{\textbf{Refresh Time (ms)}} \\
\cmidrule(lr){2-3} \cmidrule(lr){4-6}
\textbf{Dataset}
  & $|\mathcal{D}|$ & Aff.\,(\%)
  & TGL & STGN & Spd. \\
\midrule
MOOC     &    814 & 11.5 &      71.6 &   8.27 &    8.7$\times$ \\
WIKI     &  1,607 & 17.4 &     102.7 &  17.89 &    5.7$\times$ \\
REDDIT   &  2,442 & 22.2 &     140.4 &  31.22 &    4.5$\times$ \\
GDELT    &  1,326 &  8.0 &     185.0 &  14.71 &   12.6$\times$ \\
Stack-OF &  3,497 &  0.1 &  31,625   &  42.79 &  739$\times$ \\
\bottomrule
\end{tabular}
\end{table}

\begin{table}[t]
\centering
\caption{TGL per-batch inference stage breakdown (ms).
Sampling and embedding dominate ($>$80\%), which
StreamTGN eliminates for unaffected nodes.}
\label{tab:stage_breakdown}
\small
\setlength{\tabcolsep}{3pt}
\begin{tabular}{@{}l rrrrrr r@{}}
\toprule
\textbf{Dataset}
  & \makecell{Sam-\\ple}
  & \makecell{MFG}
  & \makecell{Mail}
  & \makecell{Em-\\bed}
  & \makecell{Pre-\\dict}
  & \makecell{Up-\\date}
  & \makecell{Total} \\
\midrule
MOOC     &  7.27 & 3.02 & 0.98 &  7.02 & 0.25 & 0.62 & 19.16 \\
WIKI     &  9.93 & 5.16 & 0.28 &  4.67 & 0.24 & 0.67 & 20.95 \\
REDDIT   & 10.21 & 3.47 & 1.10 &  8.23 & 0.74 & 1.04 & 24.79 \\
GDELT    & 10.59 & 2.00 & 0.24 &  7.13 & 0.29 & 1.65 & 21.90 \\
Stack-OF &  9.75 & 3.66 & 0.20 &  8.42 & 0.25 & 0.70 & 22.98 \\
\bottomrule
\end{tabular}
\end{table}

\begin{table}[t]
\centering
\caption{StreamTGN performance summary.
Batch: per-batch streaming inference;
Index: global embedding refresh.
$B{=}600$, RTX 4090. AP/AUC identical (lossless).}
\label{tab:unified}
\small
\setlength{\tabcolsep}{2.5pt}
\begin{tabular}{@{}l rr rrr cc@{}}
\toprule
& \multicolumn{2}{c}{\textbf{Batch (ms)}} &
  \multicolumn{3}{c}{\textbf{Index Refresh}} &
  \multicolumn{2}{c}{\textbf{Accuracy}} \\
\cmidrule(lr){2-3} \cmidrule(lr){4-6} \cmidrule(lr){7-8}
\textbf{Dataset}
  & TGL & Spd.
  & Aff.\,\% & TGL\,(ms) & Spd.
  & AP & AUC \\
\midrule
MOOC     & 19.2 & 3.4$\times$
         & 11.5 &      71.6 &    8.7$\times$
         & .994 & .997 \\
WIKI     & 21.0 & 2.8$\times$
         & 17.4 &     102.7 &    5.7$\times$
         & .973 & .976 \\
REDDIT   & 24.8 & 2.3$\times$
         & 22.2 &     140.4 &    4.5$\times$
         & .981 & .986 \\
GDELT    & 21.9 & 2.8$\times$
         &  8.0 &     185.0 &   12.6$\times$
         & .983 & .986 \\
Stack-OF & 23.0 & 3.7$\times$
         &  0.1 &  31,625   &  739$\times$
         & .979 & .971 \\
\bottomrule
\end{tabular}
\end{table}
\noindent\textbf{Per-batch inference.}
Table~\ref{tab:batch_inference} compares the per-batch inference
latency between TGL and StreamTGN.
TGL recomputes embeddings for all root nodes in each batch regardless
of whether their memory states have changed; StreamTGN identifies
the dirty root nodes---those whose memory was modified by recent
edges---and only recomputes their embeddings, reusing cached results
for the remainder.
StreamTGN achieves 2.3$\times$--3.7$\times$ speedup across all five
datasets.
Notably, the speedup on Stack-Overflow (3.7$\times$) is higher than
on REDDIT (2.3$\times$) because its larger node population results in
a lower fraction of dirty root nodes per batch.
The AP and AUC are identical between TGL and StreamTGN in all cases,
confirming that selective recomputation is lossless.
 
\noindent\textbf{Global index refresh.}
Table~\ref{tab:index_refresh} evaluates the cost of maintaining an
up-to-date embedding index over the entire graph.
After each batch of new edges, a serving system must refresh node
embeddings to answer nearest-neighbor or link prediction queries.
TGL refreshes all $|V|$ nodes; StreamTGN refreshes only the
$|\mathcal{D}|$ dirty nodes whose memory states were updated.
The speedup is determined by the affected ratio
$|\mathcal{D}|/|V|$: on small graphs (WIKI, 17.4\% affected),
StreamTGN achieves 5.7$\times$; on Stack-Overflow (0.14\% affected),
it achieves 739$\times$.
This demonstrates that StreamTGN's advantage scales with graph
size---on million-node graphs, fewer than 0.2\% of nodes are
affected per batch, yielding three orders of magnitude speedup.
The index refresh speedup represents StreamTGN's primary
contribution, as global refresh is the dominant cost in production
serving systems that must maintain fresh embeddings for all nodes.
 
\noindent\textbf{Stage-level breakdown.}
Table~\ref{tab:stage_breakdown} decomposes TGL's per-batch inference
pipeline into six stages.
Sampling and embedding generation together account for over 80\%
of total latency across all datasets---these are precisely the
stages that StreamTGN eliminates for unaffected nodes.
The remaining stages (predict and update) are lightweight
(typically $<$10\% combined) and must be executed regardless of
the refresh strategy.
This breakdown explains why StreamTGN's batch-level speedup
(2.3$\times$--3.7$\times$) is moderate: even dirty nodes must
still pass through the full pipeline, and the fixed-cost stages
(predict, update) cannot be skipped.
In contrast, the index refresh speedup (4.5$\times$--739$\times$)
is much larger because it operates over the entire node set,
where the ratio of skippable to total work is governed by the
affected ratio rather than the per-batch root node composition.
 
\noindent\textbf{Summary.}
Table~\ref{tab:unified} consolidates the key results.
StreamTGN provides two complementary speedups: a moderate
batch-level speedup (2.3$\times$--3.7$\times$) that benefits
real-time per-query inference, and a substantial index refresh
speedup (4.5$\times$--739$\times$) that benefits systems
maintaining up-to-date embeddings for serving.
Both speedups are achieved with zero accuracy degradation, as
StreamTGN recomputes the exact same embeddings for affected
nodes without any approximation.
\subsection{Parameter Sensitivity and System Analysis}
\label{sec:eval_sensitivity}
 Since StreamTGN recomputes the exact same embeddings for all affected nodes---differing from TGL only in \emph{which} nodes are recomputed---the prediction accuracy is identical to full
recomputation (as confirmed in Table~\ref{tab:model-comparison}).
We therefore focus on how key parameters affect StreamTGN's speedup
and where the computation time is spent.
Figure~\ref{fig:ablation} reports the index speedup and affected ratio
under three parameter sweeps across five datasets;
Table~\ref{tab:efficiency} profiles the per-stage latency breakdown.
\begin{figure*}[t]
\centering
\includegraphics[width=\textwidth]{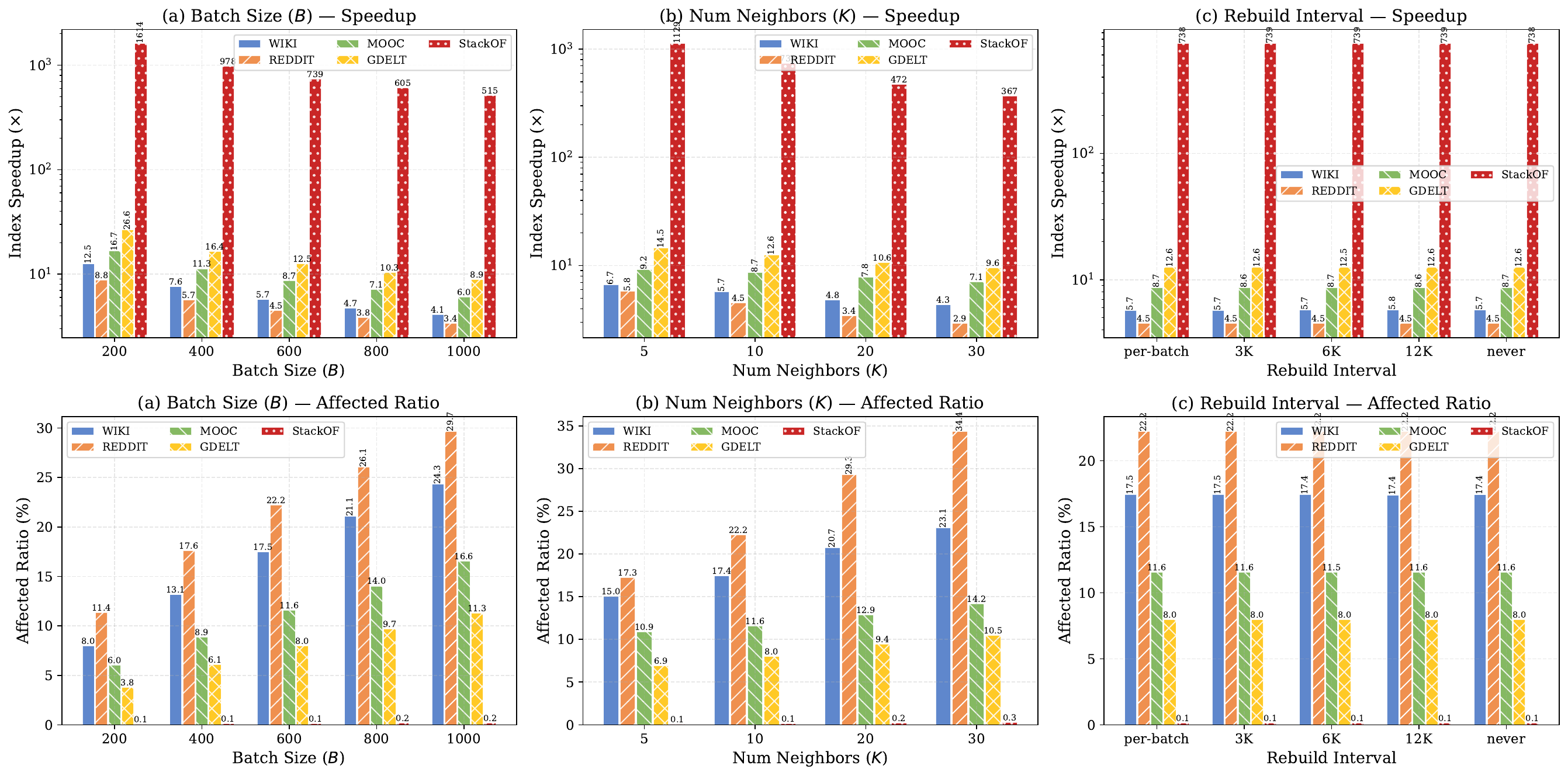}
\caption{Parameter sensitivity analysis across five datasets.
Top row: index refresh speedup (log scale);
bottom row: affected ratio (\%).
(a)~Batch size $B$: larger batches dirty more nodes, reducing
speedup from 12.5$\times$ to 4.1$\times$ on WIKI, while
Stack-Overflow maintains 515$\times$+ due to its low affected
ratio ($<$0.2\%).
(b)~Neighbor count $K$: more neighbors expand the affected set,
with similar inverse trend.
(c)~Rebuild interval: both speedup and affected ratio remain
stable from per-batch to never, confirming that periodic
rebuilds are unnecessary in practice.
Speedup is inversely proportional to affected ratio across all
settings, validating $\text{Speedup} \approx |V|/|\mathcal{A}|$.}
\label{fig:ablation}
\end{figure*}
\noindent\textbf{Batch size ($B$).}
Figures~\ref{fig:ablation}(a) show that speedup decreases
monotonically as $B$ increases from 200 to 1,000.
A larger batch introduces more new edges per step, which dirties more
nodes and raises the affected ratio.
On WIKI, the affected ratio grows from 8.0\% ($B{=}200$) to 24.3\%
($B{=}1000$), reducing the speedup from 12.5$\times$ to
4.1$\times$.
On large graphs, however, the affected ratio remains negligible
regardless of $B$---Stack-Overflow stays below 0.2\% across all
batch sizes---yielding consistently high speedups
(515$\times$--1,614$\times$).
Moreover, measured latency scales linearly with $B$
(approximately $2\times$ when $B$ doubles), confirming the
$O(|\mathcal{A}| \cdot L \cdot K \cdot d^2)$ complexity bound.
 
\noindent\textbf{Number of neighbors ($K$).}
Figures~\ref{fig:ablation}(b) show a similar inverse trend: increasing
$K$ from 5 to 30 enlarges each node's sampled neighborhood, which
expands the affected set.
On WIKI, the affected ratio rises from 15.0\% ($K{=}5$) to 23.1\%
($K{=}30$), and the speedup decreases from 6.8$\times$ to
4.3$\times$.
The impact is more pronounced on medium-scale graphs (REDDIT, MOOC)
where the expanded neighborhood represents a larger fraction of
the total graph.
On Stack-Overflow, the affected ratio remains near zero across all
$K$ values, sustaining speedups above 367$\times$.
 
\noindent\textbf{Rebuild interval.}
Figures~\ref{fig:ablation}(c) examine the tradeoff between index
freshness and rebuild cost.
Across all datasets, both the speedup and affected ratio remain
remarkably stable from ``per-batch'' through ``12K,'' indicating
that the dirty set does not accumulate significantly over time.
Even with the ``never'' policy (no periodic rebuild), the speedup
degrades only marginally---for example, WIKI stays at
5.7$\times$ and Stack-Overflow remains at 739$\times$.
This stability suggests that StreamTGN can operate with infrequent
or no rebuilds in practice, reducing the amortized overhead of
index maintenance.
\begin{table}[!t]
\centering
\caption{Per-stage latency breakdown of StreamTGN's incremental
inference pipeline ($B{=}200$, averaged over 20 batches).
Memory update (GRU) dominates at 53--67\%;
incremental-specific stages (affected detection + neighbor sampling)
contribute only 7--11\%.}
\label{tab:efficiency}
\renewcommand{\arraystretch}{1.05}
\setlength{\tabcolsep}{4pt}
\small
\begin{tabular}{@{}l rr rr rr r@{}}
\toprule
& \multicolumn{2}{c}{\makecell{Affected Det.\\+ Nbr.\ Samp.}}
& \multicolumn{2}{c}{\makecell{Feature Ret.\\+ Mem.\ Read}}
& \multicolumn{2}{c}{\makecell{Embed.\\Gen.}}
& \makecell{Mem.\\Upd.} \\
\cmidrule(lr){2-3} \cmidrule(lr){4-5}
\cmidrule(lr){6-7} \cmidrule(lr){8-8}
\textbf{Dataset} & ms & \% & ms & \% & ms & \% & \% \\
\midrule
Bitcoin   & 0.52 &  6.4 & 1.36 & 16.8 & 1.32 & 16.4 & 60.5 \\
LastFM    & 0.91 & 11.1 & 0.20 &  2.4 & 1.86 & 22.7 & 63.8 \\
Wiki-Talk & 0.76 &  9.9 & 0.19 &  2.5 & 1.82 & 23.6 & 64.0 \\
Stack-OF  & 0.99 &  9.3 & 0.23 &  2.1 & 3.81 & 35.6 & 53.0 \\
GDELT     & 0.50 &  6.8 & 0.95 & 12.8 & 1.00 & 13.5 & 66.9 \\
\bottomrule
\end{tabular}
\end{table}
\noindent\textbf{Pipeline breakdown.}
Table~\ref{tab:efficiency} decomposes StreamTGN's per-batch latency
into four stages.
Memory update (GRU computation) dominates at 53--67\% of total
latency---the same operation required by all TGN implementations
including TGL, ETC, and SWIFT.
Embedding generation accounts for 13--36\%, varying with graph
structure.
Crucially, the stages unique to incremental
computation---affected node detection and neighbor
sampling---contribute only 7--11\% overhead, confirming that
StreamTGN adds minimal cost beyond what is inherently required
by the memory architecture.
This breakdown also explains why StreamTGN's speedup is robust to
parameter changes: the dominant cost (GRU update) is fixed per
affected node, so total latency is determined primarily by
$|\mathcal{A}|$, not by system-level tuning.
 
\noindent\textbf{Summary.}
The results reveal two consistent patterns.
First, speedup is inversely proportional to the affected ratio,
validating the theoretical relationship
$\text{Speedup} \approx |V|/|\mathcal{A}|$.
Second, large graphs benefit disproportionately: because
$|\mathcal{A}|/|V|$ is inherently small on graphs with millions of
nodes, StreamTGN delivers order-of-magnitude speedups that are
robust to parameter choices.
\section{Related Work}
\label{sec:related}

We review related work along three dimensions: temporal GNN models,
training system optimizations, and inference optimizations for graph
neural networks.

\subsection{Temporal Graph Neural Networks}
\label{subsec:rw_tgnn}

Temporal GNN models extend static GNNs to dynamic graphs by
incorporating temporal information into message passing.
TGAT~\cite{xu2020inductive} applies self-attention over
time-stamped neighbors with temporal encoding.
TGN~\cite{rossi2020temporal} introduces a memory module that
maintains per-node state vectors updated via GRU cells, capturing
long-term temporal patterns beyond the immediate neighborhood.
DySAT~\cite{sankar2020dysat} uses structural and temporal
self-attention over graph snapshots.
JODIE~\cite{kumar2019predicting} and DyRep~\cite{trivedi2019dyrep}
model evolving node representations through coupled recurrent
networks.
APAN~\cite{wang2021apan} proposes asynchronous propagation to
reduce redundant computation during training.
These models define the computational patterns that training and
inference systems must support; StreamTGN is designed to be
model-agnostic and currently supports TGN, TGAT, and DySAT.

\subsection{Training System Optimizations}
\label{subsec:rw_training}

A series of systems have been proposed to accelerate temporal GNN
\emph{training} on large-scale dynamic graphs.
TGL~\cite{zhou2022tgl} designed the first unified framework with a
Temporal-CSR data structure and parallel temporal sampler, achieving
13$\times$ training speedup over individual model implementations.
Orca~\cite{li2023orca} reuses historical embeddings during training
with theoretical convergence guarantees, achieving 2--4$\times$
speedup over TGL.
Zebra~\cite{li2023zebra} replaces standard neighborhood aggregation
with temporal personalized PageRank, reducing the computation graph
while preserving accuracy.
ETC~\cite{gao2024etc} introduces adaptive batching to enlarge
training batches without exceeding information loss bounds, and a
three-step data access policy that eliminates redundant CPU--GPU
data transfers, achieving 1.6--3.3$\times$ speedup over TGL.
SIMPLE~\cite{gao2024simple} proposes dynamic GPU data placement that
caches frequently accessed features in GPU memory, reducing data
loading cost by 80--97\% and achieving 1.8--3.8$\times$ training
speedup.
SWIFT~\cite{guo2025swift} develops a secondary-memory-based pipeline
that distributes data across GPU, main memory, and disk, achieving
up to 4.3$\times$ speedup with 7.9$\times$ memory reduction.

All of these systems share a common limitation: they optimize the
training loop (backward pass, gradient synchronization, data loading)
but do not modify the inference pipeline.
At serving time, they all execute the same full-recomputation
procedure as TGL---recomputing embeddings for all $|V|$ nodes per
batch regardless of how few nodes are actually affected.
StreamTGN addresses this complementary bottleneck: it is the first
system to optimize the \emph{inference} phase of temporal GNNs, and
is orthogonal to all training-phase systems above.

\subsection{Inference Optimizations for Graph Neural Networks}
\label{subsec:rw_inference}

Inference optimization has been studied for \emph{static} GNNs but
remains largely unexplored for temporal GNNs.

\noindent\textbf{Static GNN inference.}
Several works accelerate static GNN serving through caching and
incremental computation.
GNNAutoScale~\cite{fey2021gnnautoscale} maintains historical
embeddings and updates only a mini-batch of nodes per forward pass,
enabling training and inference on graphs that exceed GPU memory.
LazyGNN~\cite{xue2023lazygnn} caches intermediate embeddings and
selectively recomputes stale entries, reducing inference cost by
avoiding redundant neighborhood aggregation.
IGLU~\cite{narayananiglu} proposes instant graph learning for
dynamic node classification by maintaining approximate embeddings
that are updated incrementally.
GAS~\cite{frasca2020sign} precomputes multi-hop aggregations to
enable scalable inference without neighbor sampling.
These methods exploit the fact that, in static or slowly evolving
graphs, most node embeddings remain stable across prediction
requests.

\noindent\textbf{Dynamic/temporal GNN inference.}
For temporal GNNs, inference optimization is significantly more
challenging because the memory module introduces sequential
dependencies: each edge updates the memory states of its endpoints,
which may in turn affect the embeddings of their neighbors in
subsequent predictions.
To our knowledge, no prior system provides incremental inference
for continuous-time temporal GNNs with formal correctness
guarantees.
The closest work is DistTGL~\cite{zhou2023disttgl}, which
distributes TGL's computation across multiple GPUs but still
performs full recomputation within each partition.
GNNFlow~\cite{zhong2023gnnflow} supports continuous temporal graph
learning on multi-GPU machines but focuses on training throughput
rather than inference latency.

StreamTGN fills this gap by providing the first incremental inference
system for temporal GNNs.
Unlike static GNN caching approaches, StreamTGN must handle the
cascading dirty-flag propagation caused by temporal memory updates
---a challenge absent in memoryless architectures.
Unlike distributed training systems, StreamTGN targets per-batch
inference latency on a single GPU, which is the deployment scenario
for most real-time serving applications.
\section{Conclusion}
\label{sec:conclusion}

Existing temporal graph neural network systems focus on accelerating
training while leaving the inference pipeline unchanged---every new
edge triggers $O(|V|)$ recomputation even though only a small
fraction of nodes are affected.
This paper presents \textbf{StreamTGN}, a streaming inference system
that exploits the inherent locality of temporal graph updates.
StreamTGN maintains persistent GPU-resident node memory and uses
lightweight dirty-flag propagation to identify the affected set
$\mathcal{A}$ after each batch of new edges, recomputing embeddings
only for $\mathcal{A}$ at $O(|\mathcal{A}|)$ cost while producing
results identical to full recomputation.
A drift-aware adaptive rebuild mechanism triggers rebuilds only
when accumulated approximation error exceeds a provable bound,
and batched streaming with relaxed ordering improves throughput
by processing edges in parallel under bounded staleness.

Experiments on eight real-world temporal graphs (2K--2.6M nodes)
demonstrate that StreamTGN achieves 4.5$\times$--739$\times$
inference speedup for TGN and up to 4,207$\times$ for TGAT,
with zero accuracy degradation.
The system generalizes across three architectures (TGN, TGAT,
DySAT) and is orthogonal to training-phase optimizations:
combining SWIFT for training with StreamTGN for inference yields
up to 24$\times$ end-to-end speedup.

For future work, we plan to extend StreamTGN to distributed
multi-GPU settings for billion-scale graphs and to explore
predictive dirty-set estimation that anticipates affected nodes
before edges arrive, enabling preemptive embedding refresh for
even lower latency.
\bibliographystyle{ACM-Reference-Format}
\bibliography{Graph_ISO}
\end{document}